\documentclass[11pt,a4paper]{article}
\usepackage{lineno}
\usepackage{longtable}
\usepackage[utf8]{inputenc}
\usepackage{amsmath}
\usepackage{amsthm}
\usepackage{algorithm}
\usepackage{algorithmic}
\usepackage[bitstream-charter]{mathdesign}
\usepackage{siunitx}
\usepackage{adjustbox}
\newtheorem{theorem}{Theorem}
\usepackage{xcolor}
\usepackage{graphicx}
\usepackage[margin=1in]{geometry}
\usepackage{float}
\usepackage{verbatimbox}
\usepackage{booktabs}
\usepackage{authblk}
\usepackage{lipsum}
\newcommand{\dis }{\displaystyle}
\usepackage{caption} 
\usepackage{subcaption} 
\usepackage[T1]{fontenc}
\theoremstyle{definition}

\renewcommand{\arraystretch}{1.5}

\makeatletter
\def\thm@space@setup{\thm@preskip=1.2\parskip \thm@postskip=0pt}
\makeatother
\usepackage{tikz}
\usetikzlibrary{arrows}
\usepackage[toc,page]{appendix}
\usepackage{xurl} 
\usepackage[authoryear, round]{natbib}

\usepackage{hyperref}\hypersetup{
	colorlinks=true,
	citecolor=blue,
	linkcolor=red,
	filecolor=blue,
	urlcolor=blue,
}

\begin{document}
	\pagenumbering{gobble} 
	
%
%

%
	
	\newpage
	
	\clearpage 
	\pagenumbering{arabic} 
	\setcounter{page}{1} 

	\title{A temperature-driven diffusion model of Usutu virus spread in Germany with spillover into neighbouring countries}

\author{
	Pride Duve$^{a}$\thanks{Corresponding author: \texttt{pride.duve@bnitm.de}},
	\;D\'aniel Cadar$^{a}$,
	\;Norbert Becker$^{b,c}$,
	\;Jonas Schmidt-Chanasit$^{a,d}$,
	\;Felix Gregor Sauer$^{a}$,
	\;and Renke L\"uhken$^{a}$\\
	
	\small $^{a}$\textit{Bernhard Nocht Institute for Tropical Medicine, Bernhard-Nocht-Straße 74, 20359 Hamburg, Germany}\\
	\small $^{b}$\textit{Institute for Dipterology, Georg-Peter-S\"u\ss-Stra\ss e 3, 67346 Speyer, Germany}\\
	\small $^{c}$\textit{University of Heidelberg, Grabengasse 1, 69117 Heidelberg, Germany}\\
	\small $^{d}$\textit{University of Hamburg, Faculty of Mathematics, Informatics and Natural Sciences, Mittelweg 177, 20148 Hamburg, Germany}
}

	\date{}
	\maketitle

	\begin{abstract}
		
Usutu virus (USUV) is a flavivirus of the Japanese encephalitis complex transmitted between \textit{Culex} mosquitoes and birds, a transmission pattern similar to that of the West Nile virus (WNV). In Germany, the first case of USUV was detected in 2010 in mosquitoes collected in the town of Weinheim, and by 2018 the virus had spread to almost the entire country. Interestingly, the infection front exhibited a clockwise rotational spread pattern throughout the years, a pattern completely different from that of the WNV. This clockwise progression corresponded closely with the spatial temperature gradient, suggesting that warmer regions probably facilitated faster viral amplification and onward transmission. Understanding the drivers that influence the spreading patterns of arboviruses is important as it guides surveillance and implementation of control strategies. In this study, we develop a reaction-diffusion partial differential equation (PDE) model to investigate the spatial spread of USUV in Germany within an extended domain that includes some neighbouring countries (Belgium, the Netherlands, and Luxembourg), thereby capturing cross-border transmission processes. Mosquito parameters, i.e., extrinsic incubation rate, mortality and biting rates, are temperature-driven, as temperature plays an important role in the activity of mosquitoes. Our model qualitatively reproduced the main spatial trends of USUV in Germany and surrounding countries. The heterogeneous spread pattern arises from the interplay of diffusion and spatially varying temperature, which together may influence determine regions with higher transmission potential.
\vspace{2em}

\end{abstract}
	
\section{Introduction}
		
	Mosquito-borne arboviruses are a global health concern. Despite their main relevance in tropical countries, Europe has recorded an increased number of mosquito-borne disease cases, both imported and autochthonous \citep{Calzolari2016}. Factors such as climate and land-use change, or long-distance travel, among others, have been closely linked to the expansion of several medically relevant mosquito species in the European continent \citep{Cheng2018, Kolimenakis2021, Mangili2005, Semenza2021}. Among the tropical viruses in Europe, the Usutu virus (USUV) has stood out in recent years as a particularly lethal pathogen for birds, especially blackbirds (\textit{Turdus merula}) \citep{Lhken2017}. USUV is a mosquito-borne flavivirus of the \textit{Japanese encephalitis} complex and circulates between birds and transmitted by mosquitoes, especially of the \textit{Culex} genus \citep{DeMadrid1974, jost2011isolation, Lhken2017}. European blackbirds have been identified as the most susceptible amplifying host, while humans, equids, and other mammals are considered dead-end hosts, a transmission cycle similar to that of West Nile virus (WNV) \citep{Nikolay2015}. 
	
	USUV was first isolated from \textit{Cx. neavei} mosquitoes in South Africa in 1959 near the Usutu river \citep{woodall1964viruses}. More than five decades later, the virus was retrospectively detected to have caused an extensive die-off of common blackbirds in the Tuscany region of Italy, in 1996 \citep{Weissenbck2013}. Since then, USUV has been confirmed in Austria, Hungary, Germany, the United Kingdom and other European countries \citep{Simonin2024}. 
	
	In Germany, the first case of USUV in mosquitoes was isolated from a pool of \textit{Cx. pipiens} biotype \textit{pipiens} trapped in August 2010 in South West Germany \citep{jost2011isolation}. The phylogenetic analysis of this strain showed a close relationship with the USUV strain from Italy \citep{Engel2016}, indicating the spread within the European continent. Soon after, Germany experienced high blackbird mortality events \citep{Becker2012, Cadar2017}, and at least two cases of human infections were recorded among healthy blood donors \citep{Allering2012, cadar2017blood}. Infected blackbirds often show neurological clinical signs such as overturning, pedalling, and in-coordination \citep{Musto2022}. On the other hand, infected humans may exhibit mild symptoms ranging from rash, jaundice, fever, and many more \citep{Cadar2022, Roesch2019}. However, the emergence of neuroinvasive USUV disease has been also documented, including a recent fatal autochthonous case in an immunocompromised patient in Hungary in 2025 \citep{Tth2025, Szab2025}.

To date, USUV remains a major concern, with its re-emergence continuously threatening the avian population, and a potential to cause disease to humans \citep{Vilibiavlek2025}. Despite its growing public health and veterinary importance, USUV remains understudied compared to other arboviruses \citep{Cheng2018}. Its genetic similarity with WNV and co-circulation in the same ecological niche further complicates its diagnosis and can have negative effects on its epidemiology \citep{Nikolay2015}. Thus, a further understanding of its spreading pattern, climate drivers, and spillover risks to mammals is required \citep{Lhken2017}. Moreover, a deeper understanding of the underlying factors influencing transmission patterns and the geographical spread of USUV can help to develop effective early warning systems. Therefore, in this study, we formulate and study a PDE model for USUV with an aim of mechanistically linking local transmission processes with spatial movement of mosquitoes and birds, and environmental heterogeneity via temperature gradients, in order to explain and model the spread of USUV.

Although the primary focus of this study is Germany, the inclusion of neighbouring regions such as the Netherlands, Belgium, and Luxembourg helps us better capture cross-border transmission dynamics, particularly in border areas where reported cases indicate possible spillover driven by ecological connectivity and host movement.
	
Several models have been studied to understand the transmission dynamics of USUV. \cite{Rubel2008} developed and calibrated the first model to explain USUV dynamics in Austria. The compartmental model consists of bird and mosquito populations, with temperature-dependent mosquito parameters. This model draws insights on how temperature patterns can be coupled with traditional process-based models to explain the USUV transmission cycle. Since then, many models have been studied, extended from this initial model considering further epidemiological and ecological factors. \cite{Reiczigel2009} extended this model into a hierarchical Bayesian model which adds stochasticity to deterministic models. The stochastic version improved the fitting to real-world USUV data, revealing the structure of inter-dependencies between different model parameters.

 \cite{Brugger2009} modelled USUV using five global climate models with temperature-dependent parameters. The model considered four different climate-warming scenarios defined by the Intergovernmental Panel on Climate Change, IPCC (20 different model-scenario combinations). Results from the simulation of the worst-case scenarios identified an endemic equilibrium with a decline of the blackbird population of about 24\%. Machine learning models such as Maxent, are well known for their high predictive accuracy, especially when there is presence-only data. Therefore, \cite{Cheng2018} compared an environmental niche model with the compartmental model by \cite{Rubel2008} to evaluate the risk of USUV circulation in Europe. The study showed that the mechanistic model was able to capture the transmission potential but could overestimate the local transmission risk. On the other hand, the environmental niche models reproduced observed patterns yet underestimated emergence risk for new, previously uninfected regions. \cite{Lhken2017} studied a boosted regression tree-based model to identify areas suitable for USUV circulation and the effects of the virus on breeding bird populations in Germany. This modelling approach identified the southwestern region of Germany as the high-risk zone of USUV, although the virus circulation affected all federal states by the year 2018 \citep{Michel2019}.
	
While USUV models have been useful in identifying risk areas and potential outbreak drivers, such as temperature and rainfall, key uncertainties remain. Most existing models focus on the fixed spatial-temporal distribution of USUV. In fact, some models only explain the spatial distribution pattern in specific years without considering movement of USUV-infected specimens in time and space. However, in reality, hosts and vectors are not stationary but move locally roughly random manner in search of food, shelter and other dispersal factors \citep{Hamer2014, Michel2019}. Mathematically, this movement can be described by the diffusion process, represented by the Laplacian operator $(\Delta)$ added to each equation, which causes the population densities to flow from initially infected areas to more suitable ones across the domain. As a result, such a model becomes a reaction-diffusion PDE model, which can better depict real spatial patterns of USUV. Therefore, in this study, we extend the ordinary differential equation based model originally developed by \cite{Rubel2008} for the dynamics of USUV in Austria and we formulate a reaction-diffusion PDE model that includes bird and mosquito populations. This model incorporates the diffusion process, to simulate the random movement of birds and mosquitoes, while the reaction terms describe their interactions at each point in time and space.

\section{Model description}
	
In this study, we adopted the ordinary differential equation model by \cite{Rubel2008}. This model is extended to include the diffusion process representing the random movement of both mosquitoes and birds, and mosquito parameters depend on the spatial temperature data for Germany, Netherlands, Belgium and Luxembourg. Our extended version consists of a system of PDEs with populations of mosquitoes and birds (Figure \ref{f1} and system \ref{xa}). The mosquito population at any given time $t$ and location $\boldsymbol{x}\in (x,y)$ is divided into susceptible: $S_V,$ exposed: $E_V,$ and infectious mosquitoes: $I_V,$ with a total population $N_V=S_V+E_V+I_V.$ Mosquitoes are recruited into the susceptible class via a logistic recruitment term: $\dis b_VN_V\left[1-\frac{N_V}{K_V}\right].$ Susceptible mosquitoes that interact with infectious birds progress to the exposed class, with a temperature dependent mosquito biting rate:
	
	\begin{equation}\label{e1}
		\beta(T,\boldsymbol{x})=\frac{0.344}{\dis 1+1.231e^{\dis -0.184(T-20)}},
	\end{equation}
	
	and a probability of acquiring an infection from a successful bite of an infectious bird of $p_{v}.$ The force of infection on mosquitoes is given by 
	
	\begin{equation}\label{e2}
		\lambda_{BV}(T,\boldsymbol{x})=\frac{p_{v}\beta(T,\boldsymbol{x})I_B}{N_B}.
	\end{equation}
	
	In the exposed class, mosquitoes are infected but not yet infectious. Only after an incubation period of $\dis \frac{1}{\gamma_V(T,\boldsymbol{x})}$ days they become infectious at a latency rate:
	
	\begin{equation}\label{e3}
		\gamma_{V}(T,\boldsymbol{x})= \dis \frac{\dis 1}{\dis e^{\dis -0.09T+5.36}},
	\end{equation}
	
	which was fitted for WNV in \textit{Cx. pipiens} by \cite{Heidecke2024}. A WNV latency rate is adopted because suitable USUV data to estimate the USUV extrinsic incubation period is unavailable. However, given the similarities between the two viruses, this assumption is reasonable \citep{Rubel2008}. Mosquitoes in all health states die at a natural death rate: 
	
	\begin{equation}\label{e4}
		\mu_V(T,\boldsymbol{x})= \frac{0.0025T^2-0.094T+1.0257}{10},
	\end{equation}
	fitted in \cite{Rubel2008}.

	The population of birds under study is divided into susceptible: $S_B,$ exposed: $E_B,$ infectious: $I_B,$ recovered: $R_B,$ and dead birds: $D_B.$ Birds are recruited into the susceptible class through a logistic recruitment rate of $\dis b_BN_B\left[1-\frac{N_B}{K_B}\right],$ where $N_B=S_B+E_B+I_B+R_B$ is the total bird population. With a force of infection
	
	\begin{equation}
		\lambda_{VB}(T,\boldsymbol{x})=\frac{\phi_{B}(\boldsymbol{x})p_{b}(\boldsymbol{x})\beta(T,\boldsymbol{x})I_V}{N_V}
	\end{equation}
	
 susceptible birds $(S_B)$ move to the exposed class $(E_B)$. The parameter $p_b$ measures the probability that a successful bite by an infectious mosquito leads to a new bird infection, while $\beta(T,\boldsymbol{x})$ is the mosquito biting rate and $\phi_B$ is the mosquito to bird ratio. Similar to WNV models \citep{Bhowmick2020, Bhowmick2023, Laperriere2011, Mbaoma2024}, the mosquito to bird ratio is an important factor that determines the per biting pressure each bird receives, and therefore how efficiently a virus can spread.
	
Birds leave the exposed class $(E_B)$ at a latency rate of $\gamma_B,$ where they become infectious. A proportion $\nu_B$ of infectious birds $(I_B)$ die due to the USUV at a rate $\alpha_B,$ while a proportion $1-\nu_B$ recovers at the same rate. The natural mortality rate of birds at all life stages is denoted by $\mu_B.$ The schematic diagram of the model under study is shown in Figure (\ref{f1}), while the corresponding system of equations is described by system (\ref{xa}).
	
	\begin{figure}[H]
		\centering
		\includegraphics[width=1\textwidth]{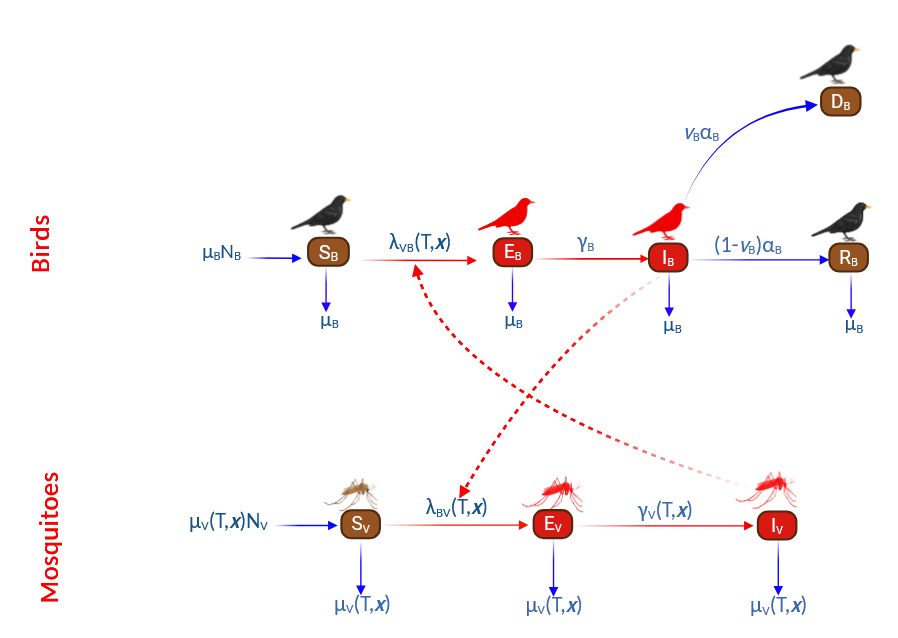}
		\caption{Flow chart diagram for the transmission dynamics of USUV. Red colours indicate infected classes, while dotted lines indicate cross-infection between birds and mosquitoes. Figure created using BioRender. Arbovirologie (2025), https://BioRender.com/r9wek95. Content not licensed under the Creative Commons Attribution (CC BY) license.}
		\label{f1}
	\end{figure}
	
	The system of equations becomes:
	\begin{equation} 
		\begin{split}
			\frac{\partial S_V}{\partial t} & \;=\; D_{1}\Delta S_V \quad+ \quad b_VN_V\left[1-\frac{N_V}{K_V}\right]-[\lambda_{BV}(T,\boldsymbol{x})+\mu_V(T,\boldsymbol{x})]S_V,\\
			\frac{\partial E_V}{\partial t} &\; =\; D_{1}\Delta E_V \quad+ \quad \lambda_{BV}(T,\boldsymbol{x})S_V-[\gamma_V(T,\boldsymbol{x})+\mu_V(T,\boldsymbol{x})]E_V,\\
			\frac{\partial I_V}{\partial t} &\; =\; D_{1}\Delta I_V \quad+\quad \gamma_V(T,\boldsymbol{x})E_V-\mu_V(T,\boldsymbol{x})I_V,\\
			\frac{\partial S_B}{\partial t} & \;=\; D_{2}\Delta S_B \quad +\quad b_BN_B\left[1-\frac{N_B}{K_B}\right]-[\lambda_{VB}(T,\boldsymbol{x})+\mu_B]S_B,\\
			\frac{\partial E_B}{\partial t} & \;=\; D_{2}\Delta E_B \quad +\quad \lambda_{VB}(T,\boldsymbol{x})S_B-[\gamma_B+\mu_B]E_B,\\
			\frac{\partial I_B}{\partial t} & \;= \; D_{2}\Delta I_B \quad + \quad \gamma_BE_B-[\alpha_B+\mu_B]I_B,\\
			\frac{\partial R_B}{\partial t} &\; = \; D_{2}\Delta R_B \quad +\quad (1-\nu_B)\alpha_BI_B-\mu_BR_B,\\
			\frac{\partial D_B}{\partial t} & \;=\;  \quad\quad\quad\quad\quad\quad\quad\nu_B\alpha_BI_B,
		\end{split}\label{xa}
	\end{equation}
	
	subject to the initial conditions: $S_V(0,x)=\psi_1(x)$, $E_V(0,x)=\psi_2(x)$, $I_V(0,x)=\psi_3(x)$, $S_B(0,x)=\psi_4(x),$ $E_B(0,x)=\psi_5(x)$, $I_B(0,x)=\psi_6(x)$  $R_B(0,x)=\psi_7(x),$ $D_B(0,x)=D_{B_0},$ where $\displaystyle \boldsymbol{x}=(x,y)\in\Omega,$ and $\displaystyle \Delta =\frac{\partial^2(\cdot)}{\partial x^2}+\frac{\partial^2(\cdot)}{\partial y^2}.$ We assume the system moves in a region $\displaystyle \Omega\subset \mathbb{R}^2,$ with a smooth boundary $\partial\Omega$ according to Fick's law \citep{Fick1855}, so that in the initial conditions, $\boldsymbol{x}\in\Omega,$ where $\psi_i\in C^2(\Omega)\cap C(\bar{\Omega})$ and subject to the homogeneous Neumann boundary conditions:
	
	$$\frac{\partial S_V}{\partial \nu}=\frac{\partial E_V}{\partial \nu}=\frac{\partial I_V}{\partial \nu}=\frac{\partial S_{B}}{\partial \nu}=\frac{\partial E_{B}}{\partial \nu}=\frac{\partial I_{B}}{\partial \nu}=\frac{\partial R_{B}}{\partial \nu}=0,$$
	
	$\displaystyle \text{for}\; \boldsymbol{x}\in\partial\Omega,\;t>0,$ and $\nu$ is the unit outer normal to $\Omega.$ Variables and parameters used in the model are summarized in Table (\ref{t02}).

Detailed proofs of the existence, positivity, and boundedness of solutions to the PDE model are presented in Appendix~(\ref{mt}).
	
\section{Methods}

\subsection{Temperature forcing}

Temperature is a key environmental variable for USUV transmission and spread \citep{Rubel2008}. It is known for its influence on the mosquito life-history traits, including biting rate, development time, survival, and dispersal. In addition, viral processes such as the extrinsic incubation period within the mosquito shorten at higher temperatures, increasing transmission potential. Consequently, temperature directly modulates the effective transmission rate and can create seasonal windows when outbreaks are possible, especially in temperate regions such as Germany. In this study, temperature-dependent parameters, i.e., mosquito biting rate, latency rate and mortality rate, are adopted from experimental studies \citep{Rubel2008} and \citep{Heidecke2024}. The thermal functions are defined in Eqn. (\ref{e1}, \ref{e3}, \ref{e4}) and their spatio-temporal plots in Fig. (\ref{y}). 

\begin{figure}[H]
	\centering
	\begin{subfigure}{\textwidth}
		\centering
		\includegraphics[width=1.0\textwidth]{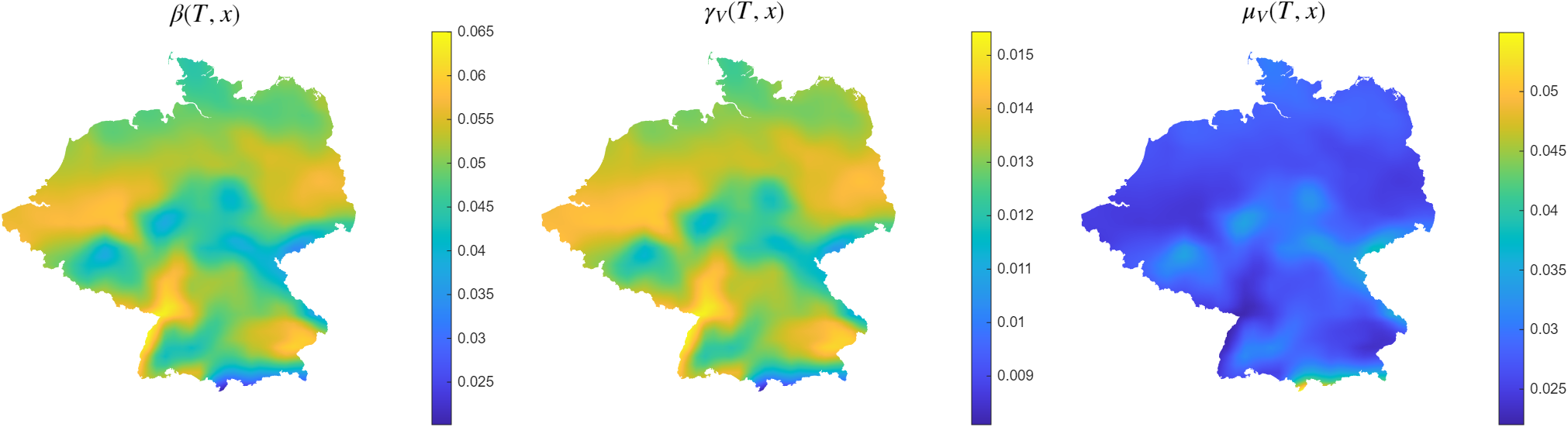}
		\caption{Space-dependent thermal mosquito parameters for the year 2018.}
	\end{subfigure}
	\hfill
	\begin{subfigure}{\textwidth}
		\centering
		\includegraphics[width=1.0\textwidth]{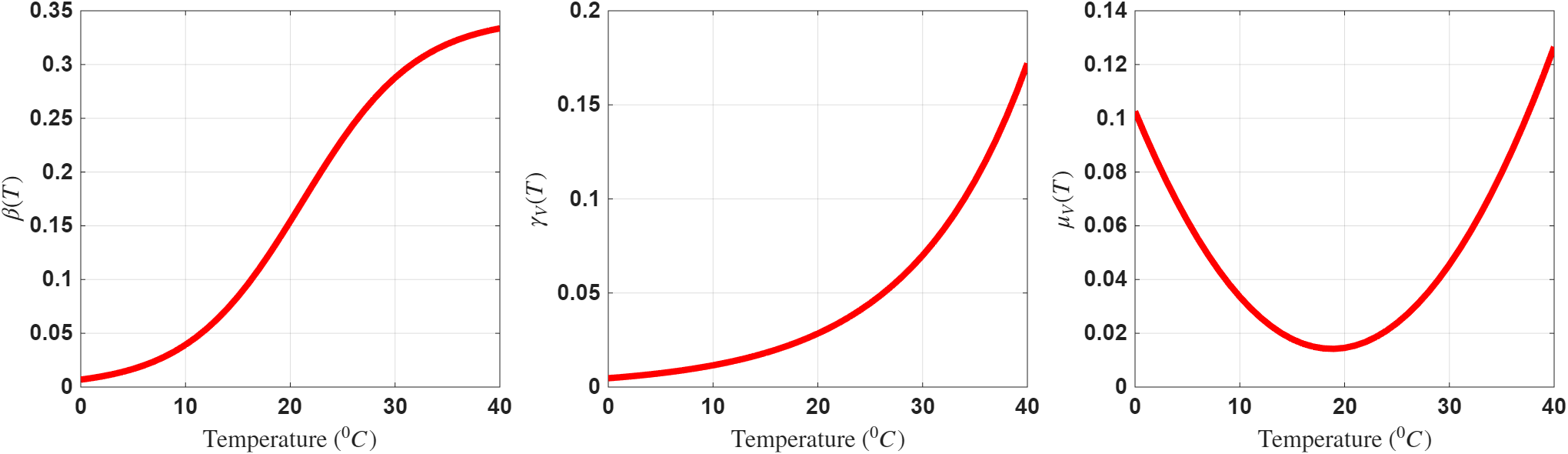}
		\caption{Temperature evolution for mosquito parameters.}
	\end{subfigure}
	\caption{\small Spatio-temporal mosquito biting rate $\beta(T,x)$, mortality rate $\mu_V(T,x)$ and extrinsic incubation rate $\gamma_V(T,x)$ over the geometry (merged maps of Belgium, Germany, Luxembourg and the Netherlands).}
	\label{y}
\end{figure}

\subsection{USUV free equilibrium point and the local reproduction number}
	
	System (\ref{xa}) admits an USUV-free equilibrium point $\displaystyle \mathcal{E}_0:$ 
	
	\begin{equation}
		\mathcal{E}_0:\quad \left[S_V^0,E_V^0,I_V^0,S_B^0,E_B^0,I_B^0,R_B^0,D_B^0\right]=\left[K_V\left[1-\frac{\mu_V}{b_V}\right],0,0,K_B\left[1-\frac{\mu_B}{b_B}\right],0,0,0,0\right].
	\end{equation}
	
	The uniqueness of $\mathcal{E}_0$ follows from [\citep{Wang2018}, Lemma 2.1]. The local basic reproductive number denoted by $R_0(\boldsymbol{x})$ is the average number of new infections produced by a single infectious mosquito or bird. In this case, we use it to create spatial USUV risk maps such that regions where $R_0(\boldsymbol{x})$ imply a higher risk, compared to regions where it is low. $R_0(\boldsymbol{x})$ is computed using the next generation matrix method defined by \citep{Diekmann2009}, and applied in an USUV model by \citep{Rubel2008}. Infected terms of system (\ref{xa}) are re-written in terms of the rate of appearance of new infections in compartment $i$, denoted $F_i,$ and the rate of transfer into and out of compartment $i,$ denoted $V_i,$ where $i=E_V, I_V, E_B, I_B$ \citep{vandenDriessche2002}, such that:

	$$\mathcal{F}(\boldsymbol{x})=\left[\setlength\arraycolsep{15pt}
	\begin{array}{c}
		\displaystyle \frac{p_v\beta(\boldsymbol{x})I_BS_V}{N_B}  \\[3.5ex]
		0 \\[3.5ex]
		\displaystyle  \frac{\phi_Bp_b\beta(\boldsymbol{x})I_VS_B}{N_V}  \\[3.5ex]
		0  \\
	\end{array}\right],\quad\text{and}\quad 
	\mathcal{V}(\boldsymbol{x})=\left[\setlength\arraycolsep{15pt}
	\begin{array}{c}
		\displaystyle 	\left(\gamma_V(\boldsymbol{x})+\mu_V(\boldsymbol{x})\right)E_V \\[3.5ex]
		\displaystyle 	-\gamma_V(\boldsymbol{x})E_V+\mu_V(\boldsymbol{x})I_V \\[3.5ex]
		\displaystyle  \left(\gamma_B+\mu_B\right)E_B \\[3.5ex]
		\displaystyle  -\gamma_BE_B+(\alpha_B+\mu_B)I_B \\
	\end{array}\right].
	$$

	Taking partial derivatives, we have: 
	
	$$F(\boldsymbol{x})=\left[\setlength\arraycolsep{4pt}
	\begin{array}{cccc}
		\displaystyle 0&0&0&\displaystyle \frac{p_v\beta(\boldsymbol{x})S_V^0}{N_B^0}  \\[3.5ex]
		0&0&0&0 \\[3.5ex]
		\displaystyle  0&\displaystyle \frac{\phi_Bp_b\beta(\boldsymbol{x})S_B^0}{N_V^0}&0&0  \\[3.5ex]
		0&0&0&0  \\
	\end{array}\right],\quad\text{and}\quad 
	V(\boldsymbol{x})=\left[\setlength\arraycolsep{2pt}
	\begin{array}{cccc}
		\gamma_V(\boldsymbol{x})+\mu_V(\boldsymbol{x}) & 0 & 0 & 0 \\[3.5ex]
		-\gamma_V & \mu_V & 0 & 0 \\[3.5ex]
		0 & 0 & \gamma_B+\mu_B & 0 \\[3.5ex]
		0 & 0 & -\gamma_B & \alpha_B+\mu_B
	\end{array}\right].
	$$
	The basic USUV reproductive number is given by the largest eigenvalue of the matrix $FV^{-1},$ and is given by 
	
	\begin{equation}
		R_0(\boldsymbol{x})= \sqrt{\left(\frac{\beta(\boldsymbol{x})p_v(\boldsymbol{x})\gamma_B(\boldsymbol{x})}{[\gamma_B(\boldsymbol{x})+\mu_B(\boldsymbol{x})][\alpha_B(\boldsymbol{x})+\mu_B(\boldsymbol{x})]}\right)\times \left(\frac{\phi_B(\boldsymbol{x})\beta(\boldsymbol{x})p_b(\boldsymbol{x})\gamma_V(\boldsymbol{x})}{[\gamma_V(\boldsymbol{x})+\mu_V(\boldsymbol{x})]\mu_V(\boldsymbol{x})}\right)}.
	\end{equation}

	The USUV reproductive number $(R_{0}(\boldsymbol{x}))$ represents the average number of new infections caused by a single infected mosquito or bird, in a completely susceptible population. In precise, $(R_{0}(\boldsymbol{x}))$ can be written as 
	
	$$R_{0}(\boldsymbol{x})=\dis \sqrt{r_b\times r_m},\; \text{where}\; r_b=\frac{\beta(\boldsymbol{x})p_v\gamma_B}{(\gamma_B+\mu_B)(\alpha_B+\mu_B)}\; \text{and}\; r_m=\frac{\beta(\boldsymbol{x})p_b\gamma_V(\boldsymbol{x})}{(\gamma_V(\boldsymbol{x})+\mu_V(\boldsymbol{x}))\mu_V(\boldsymbol{x})}$$
	
	can be interpreted as the average number of new bird infections per infectious mosquito, and the average number of new mosquito infections per infectious birds, respectively. $(R_{0}(\boldsymbol{x}))$ shows the spatial variation of USUV across different regions (Figure \ref{fb}). For a specific location in Germany, higher values of $(R_{0}(\boldsymbol{x}))$ indicate a heightened USUV risk compared to regions with lower values. The spatial heterogeneity in $(R_{0}(\boldsymbol{x}))$ further reflects temperature associated geographical variation of USUV transmission and spread.

\subsection{USUV data}

Unlike many human and animal (pet and livestock) infectious diseases, USUV circulation is not easy to detect because it primarily affects wild bird populations, which are more challenging to monitor. Human clinical surveillance is limited because infections in humans are usually asymptomatic. In some instances, infections have been identified incidentally through blood donor screening \citep{Angeloni2023}, but such detections may not reliably reflect broader infection patterns.

As a result, a nationwide surveillance programme was established in Germany, encouraging citizens to report and submit dead birds for laboratory testing. This citizen science initiative, led by the Bernhard Nocht Institute for Tropical Medicine, enabled large-scale spatial and temporal monitoring of USUV circulation nationwide. However, as with most passive surveillance systems, certain limitations need to be recognised, including potential reporting bias towards urban areas, variation in carcass condition caused by shipping delays and decomposition, and occasional inaccuracies in recorded collection locations. Despite these challenges, USUV data for Germany indicated a widespread presence across the country by 2018 (Figure \ref{fb}), with regions in the southwest (along the Upper Rhine valley) and Hamburg showing high circulation.

While the available German surveillance data largely depend on passive monitoring of dead birds, the available data from the Netherlands were collected within a comprehensive One Health active surveillance program for USUV in 2016. This nationwide monitoring system included sampling of live, free-ranging birds and dead birds until 2022. Live birds were captured by trained volunteer bird ringers across the country and sampled via throat and cloacal swabs, as well as blood, for serological testing. Dead birds were mainly reported through citizen submissions and examined at the Dutch Wildlife Health Centre, where necropsies were performed, and brain tissue was tested for USUV. From 2019, mosquitoes were also sampled for USUV, and all submitted birds underwent systematic testing regardless of the suspected cause of death \citep{Mnger2025}.

The increased USUV circulation in both birds and humans in the neighbouring Luxembourg prompted the initiation of passive surveillance as an early warning system \citep{Snoeck2022, VilibicCavlek2020}. A total of 61 samples were tested as part of surveillance based dead-bird monitoring. During the surveillance period, 33 birds were examined, and only one Eurasian blackbird (\textit{Turdus merula}) detected in September 2020 tested positive for USUV, marking the first confirmed case in the country. 

In Belgium, USUV circulation has also been documented primarily through regional passive surveillance, rather than a continuous nationwide monitoring program similar to that in Germany and the Netherlands. Following unusual bird mortality events, USUV was first detected in southern Belgium in 2016, with confirmed cases identified through RT-qPCR testing of dead wild birds submitted to wildlife rehabilitation centres \citep{Benzarti2020, Cadar2017}. Between 2017 and 2018, expanded passive monitoring detected numerous USUV-positive birds across several species, particularly Eurasian blackbirds (\textit{Turdus merula}), indicating active local transmission and suggesting the virus had become established \citep{Benzarti2020}.

Our data for fitting the model thus comes from various heterogeneous sources. More precisely, the German and Dutch data mainly originate from citizen science projects conducted in those countries. Cases from the Netherlands are published in \cite{Mnger2025} and \cite{J.M.A.vandenBrand2025}, while the single case observed in Luxembourg in 2020 has been documented in \citep{Snoeck2022}, and the Belgian cases were reported in several studies, including \citep{Cheng2018, Qiu2025}. Citizen science data from Germany are published along with this study. Consequently, the data are only intended for reference to help identify regions with higher circulation risks.

	\subsection{Simulation framework}

	In the model simulation, the densities of susceptible birds and mosquitoes are assumed to be spatially homogeneously distributed across the entire domain. An estimated area for mainland Germany, Netherlands, Belgium and Luxembourg from our simplified geometry is A=426,840 km$^2\approx 4.2684\times 10^{11}$ m$^2.$ Thus, assuming a total bird population of 110 million all over Germany, we estimate a bird density of $110\times 10^6/A\approx 2.6\times 10^{-4}$ susceptible birds per square meter. Given the high abundance of \textit{Culex} mosquitoes \citep{Becker2012c, Rudolf2013, Walther2017}, we assume an initial susceptible mosquito population of $2\times 10^9/A$ yielding a density of 0.00467 susceptible mosquitoes per square meter.
	
	Diffusion coefficients are estimated using inverse modelling, by adjusting them within the species reasonable mobility rates per day. The selected values are then interpreted using the root-mean-square displacement (RMSD), a metric that describes the dispersal distance species move from the initial source under diffusion after a time $\tau$ \citep{Einstein1905, metelmann2019development}. The RMSD is calculated using the formula:
	
	$$\text{RMSD} =\sqrt{2pD\tau},$$
	
	where p = 2 is the spatial dimension, D is the diffusion coefficient and $\tau$ is the time step. 
	
	Although reported cases likely under-represent the true extent of circulation due to a risk of under-reporting, they nevertheless serve as reliable indicators of regions where transmission is active. Because of this, we seed initial conditions for the infected classes in the Southwestern (around Weinheim) and Western regions (around Bonn), as described in the observed dataset. Infections were confirmed in these regions and their surroundings by the end of 2011, and thus the initial condition distribution is centred at these regions, according to a Gaussian distribution:

	\begin{equation}\label{r}
		f(t,\boldsymbol{x})= e^{\dis -\frac{(x-x_0)^2+(y-y_0)^2}{2\sigma^2}}.
	\end{equation}
	
	In eq (\ref{r}), $(x_0,y_0)$ are the central nodes of the infected regions, and $\sigma=40$ km is the standard deviation which controls the spatial radius over which the initial infection is present. Integrating eq (\ref{r}) over the mesh gives the total mass $B,$ which normalizes initial conditions for infected classes defined by the Gaussian. To ensure effective interaction of reaction terms during the simulation process, we choose initial infected populations in such a way that $S_V/A-E_V/B-I_V/B>0$ and $S_B/A-E_B/B-I_B/B-R_B/B>0.$ The model variables and initial conditions are summarized in Table (\ref{t02}). 
	
	Our PDE model is solved numerically using the Matlab PDEToolbox \citep{MATLAB2025}, which solves PDEs of the form:
	\begin{equation}
		m\frac{\partial^2u}{\partial t^2}+d\frac{\partial u}{\partial t}-\nabla \cdot\left(c\nabla u\right)+au=f.
	\end{equation}

	In our case, $m=0, d=1,$ $a=0,$ and matrices $c$ and $f$ are defined in Appendix (\ref{mat}). The matrix $c$ consists of the diffusion coefficients, while $f$ stores reaction terms. The full algorithm for the solver is fully explained in our West Nile virus study \citep{duve2025modeling}. It is important to compare our model results with the observed data. However, given that the PDE model produces a continuous solution, while the observed data is discrete, fitting the model is not straightforward. We fitted our model predictions by aggregating our annual PDE solution and the annual observed data into the GADM level-2 administrative units of Germany \citep{GADM}, comprising 356 rural districts (Landkreise) and urban districts (kreisfreie Städte). We then tested for a Spearman's correlation between the simulated districts and observed, by checking if high predicted regions correspond to regions with high observed cases. 
	
	\begin{table}
		\caption{\small Definition of state variables and parameters.}
		\renewcommand{\arraystretch}{1.5}
		\small
		\centering
		\begin{adjustbox}{width=1\textwidth}
			\begin{tabular}{llll}
				\toprule 	
				\textbf{Variable/Parameter} & \textbf{Definition}&\textbf{Value}& \textbf{Source} \\
				\midrule 
				$S_V$ & susceptible mosquitoes (mosquitoes/m$^2$)&$2\times 10^9/A$&defined by the authors\\ 
				$E_V$ & exposed mosquitoes (mosquitoes/m$^2$)&$10\times 10^6/B$&defined by the authors\\ 
				$I_V$ & infectious mosquitoes (mosquitoes/m$^2$)&$10\times 10^6/B$&defined by the authors\\ 		
				$S_B$ & susceptible birds (birds/m$^2$)&$110\times 10^6/A$&defined by the authors\\ 
				$E_B$ & exposed birds (birds/m$^2$)&$1\times 10^6/B$&defined by the authors\\
				$I_B$ & infectious birds (birds/m$^2$)&$1\times 10^6/B$&defined by the authors\\ 
				$R_B$ & recovered birds (birds/m$^2$)&$1\times 10^6/B$&defined by the authors\\ 
				$D_B$ & dead birds (birds/m$^2$)&$1\times 10^6/B$&Defined by the authors\\ 		
				$\beta(T,\boldsymbol{x})$ &mosquito biting rate&Eq (\ref{e1})&\citep{Rubel2008}\\ 
				$\gamma_V(T,\boldsymbol{x})$ &latency rate of mosquitoes&Eq (\ref{e3})&\citep{Heidecke2024}\\ 
				$\mu_V(T,\boldsymbol{x})$ &mortality rate of mosquitoes&Eq (\ref{e4})&\citep{Rubel2008}\\ 
				$\gamma_B(\boldsymbol{x})$ & latency rate of birds&$0.667$&\citep{Rubel2008}\\
				$\alpha_B(\boldsymbol{x})$ & removal rate of birds&$0.0182$&\citep{Rubel2008}\\
				$\nu_B(\boldsymbol{x})$ & fraction of birds dying due to the infection&$0.3$&\citep{Rubel2008}\\
				$\mu_B(\boldsymbol{x})$ & mortality rate of birds&$0.0012$&\citep{Rubel2008}\\
				$p_{v}(\boldsymbol{x})$ &probability of virus transmission by infectious birds&0.125&\citep{Rubel2008}\\ 
				$p_{b}(\boldsymbol{x})$ &probability of virus transmission by infectious mosquitoes&1.000&\citep{Rubel2008}\\  
				$\phi_{B}(\boldsymbol{x})$ &mosquito to bird ratio&Table (\ref{t03})&defined by the authors\\
				$K_V$ &environmental carrying capacity of mosquitoes&$S_V(0,\boldsymbol{x})$&\citep{Rubel2008}\\ 
				$K_B$ &environmental carrying capacity of birds&$S_B(0,\boldsymbol{x})$&\citep{Rubel2008}\\ 
				$D_1(\boldsymbol{x})$ &diffusion coefficient for mosquitoes&Table (\ref{t03})&defined by the authors\\ 
				$D_2(\boldsymbol{x})$ &diffusion coefficient for birds&Table (\ref{t03})&defined by the authors\\ 
				\bottomrule
			\end{tabular}
		\end{adjustbox}
		\label{t02}
	\end{table}

	\section{Results}
	
With USUV initially detected in Weinheim and Bonn in 2011, a notable range expansion is observed towards Belgium, the Netherlands, and northern Germany (Figure \ref{f2}). Cases increased since 2011, although the spread has been slower over the years. The root mean squared displacement (RMSD) in both mosquitoes and birds increased steadily during this period. The Spearman's correlation values, reported as $\rho$, ranged from 0.29 to 0.44 across all years (2011-2018), indicating a weak positive monotonic relationship between the observed data and the model predictions (Figure \ref{g2}). This could be because the PDE solution is continuous and highly sensitive to small values, similar to noise, which may be present in regions without confirmed cases. On the other hand, under-reporting cannot be ruled out either. However, despite this, there is some visual agreement between the observed distributions and the predicted cases, with the main circulation area remaining in the south-west (around the Rhein valley) (Figures \ref{f2}, \ref{l2}, \ref{k2}, \ref{g2}). In 2016, an isolated case was recorded in the eastern part of Germany, within Berlin, and more cases were observed in the Netherlands. Since then, a rising number of new cases has been detected in Belgium, the Netherlands, and subsequently across Germany. Luxembourg lies in a low-risk area (Figure \ref{fb}), and the first case in this country was reported in 2020.

	\begin{table}
		\caption{\small Values for the diffusion coefficients for mosquitoes and birds $(D_1, D_2)$ in km$^2$ per day, root mean squared displacement (RMSD) in km, mosquito to bird ratio $(\phi_B)$ Spearman's rank correlation coefficient $(\rho),$ and p-values.}
		\renewcommand{\arraystretch}{1.1}
		\small
		\centering
		\begin{adjustbox}{width=1.0\textwidth}
			\begin{tabular}{llllllll}
				\toprule 	
				\textbf{Year}&D$_1$&D$_2$ &RMSD$_1$&RMSD$_2$&$\phi_B$&$\rho$& p-value\\
				\midrule 
				2012& 3.0 & 6.0  &57.4  & 81.2  &30& 0.37 &$5.5\times 10^{-13}$\\
				
				2013& 8.0 & 10.0 &93.8 & 104.9 &40& 0.38 &$1.8\times 10^{-13}$\\
				
				2014& 3.0 & 6.0  &57.4 & 81.2 &40& 0.40 &$6.1\times 10^{-15}$\\
				
				2015& 6.0 & 9.0  &81.2 & 99.5 &30& 0.41 &$1.4\times 10^{-15}$\\
				
				2016& 5.0 & 8.0  &74.2 &93.8&50& 0.44 &$5.7\times 10^{-18}$\\
				
				2017& 5.0 & 8.0  &74.2 & 93.8 &30& 0.41 &$3.4\times 10^{-16}$\\
				
				2018& 5.0 & 8.0  &74.2 & 93.8 &40& 0.29 &$4.4\times 10^{-8}$\\
				\bottomrule
			\end{tabular}
		\end{adjustbox}
		\label{t03}
	\end{table}

	\begin{figure}[H]
		\centering
		\includegraphics[width=1.1\textwidth]{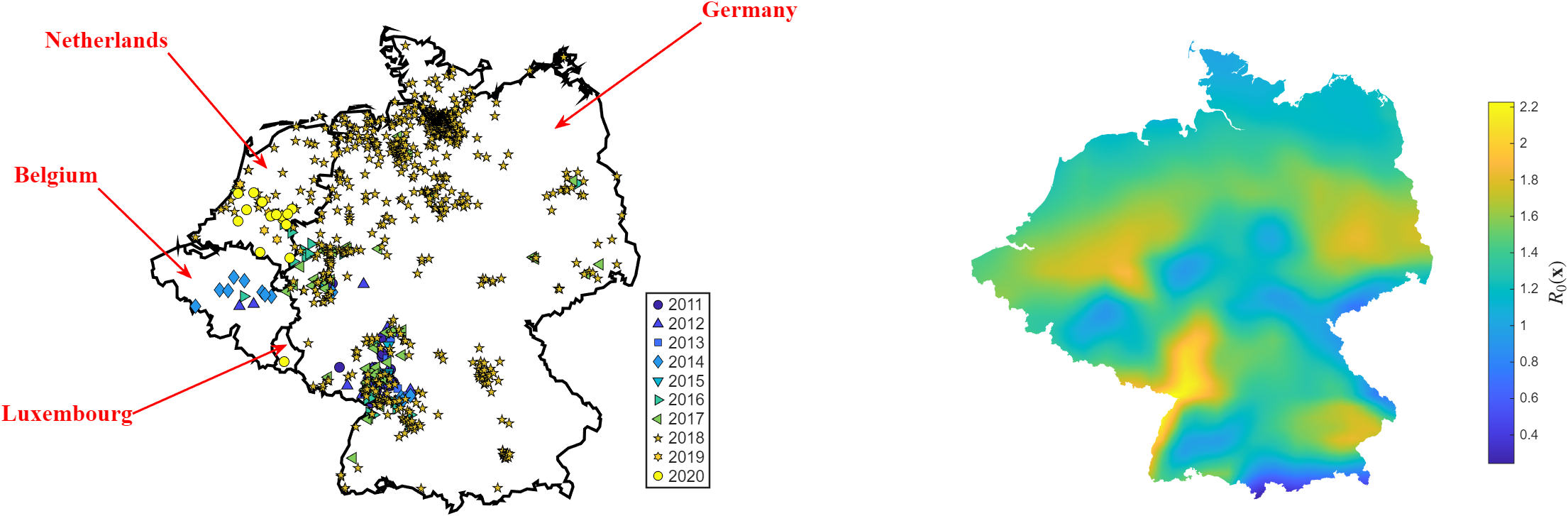}
		\caption{\small Observed data.}
		\caption{\small USUV data for the years 2011-2020 (Left) and the basic reproductive number $R_0(\boldsymbol{x})$ (Right), for the year 2018.}
		\label{fb}
	\end{figure}

	\begin{figure}[H]
		\centering
		\begin{subfigure}{0.35\textwidth}
			\centering
			\includegraphics[width=\textwidth]{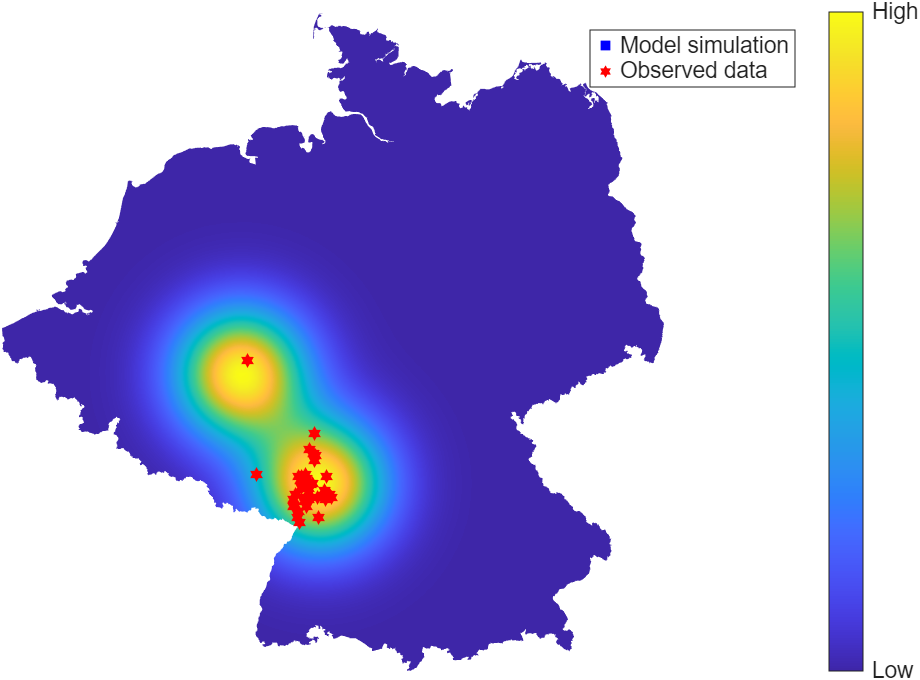}
			\caption{\small Model initialization (2011).}
			\label{f01}
		\end{subfigure}
		\hfill
		\begin{subfigure}{0.35\textwidth}
			\centering
			\includegraphics[width=\textwidth]{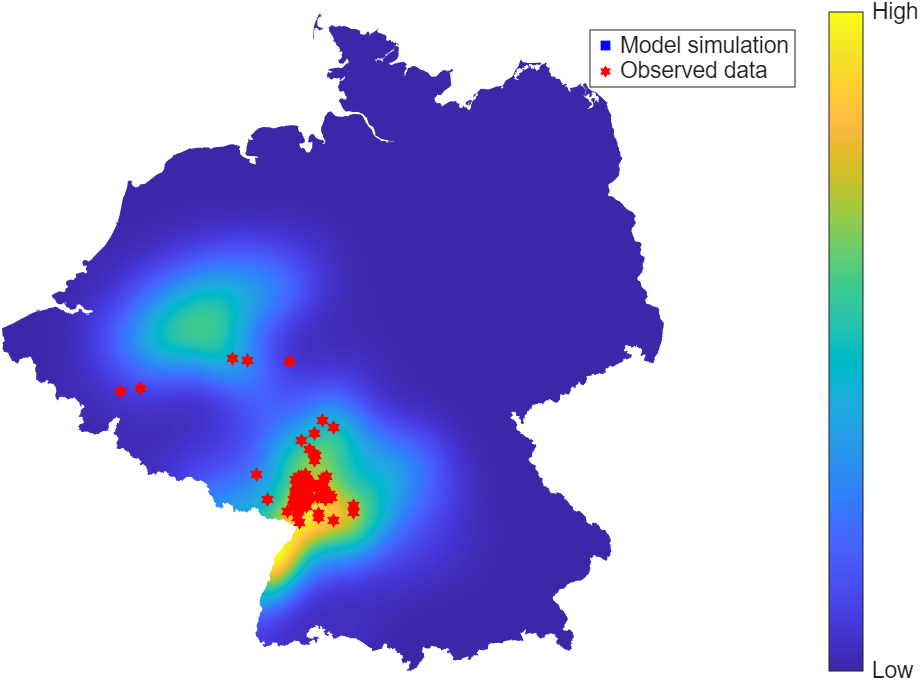}
			\caption{\small 2011-2012.}
			\label{f03}
		\end{subfigure}
		\hfill
		\begin{subfigure}{0.35\textwidth}
			\centering
			\includegraphics[width=\textwidth]{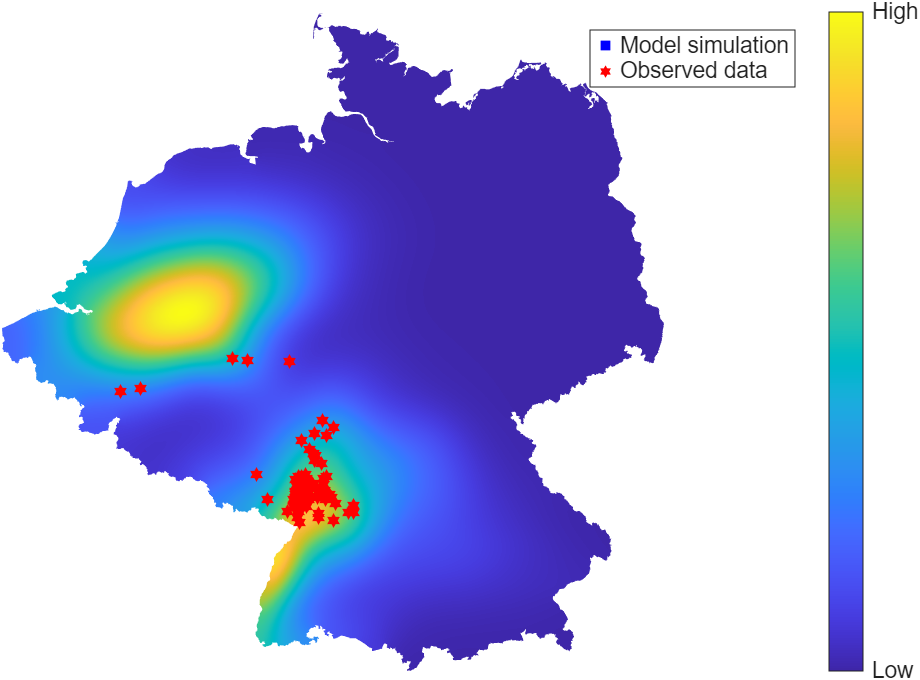}
			\caption{\small 2011-2013.}
			\label{f05}
		\end{subfigure}
		\hfill
		\begin{subfigure}{0.35\textwidth}
			\centering
			\includegraphics[width=\textwidth]{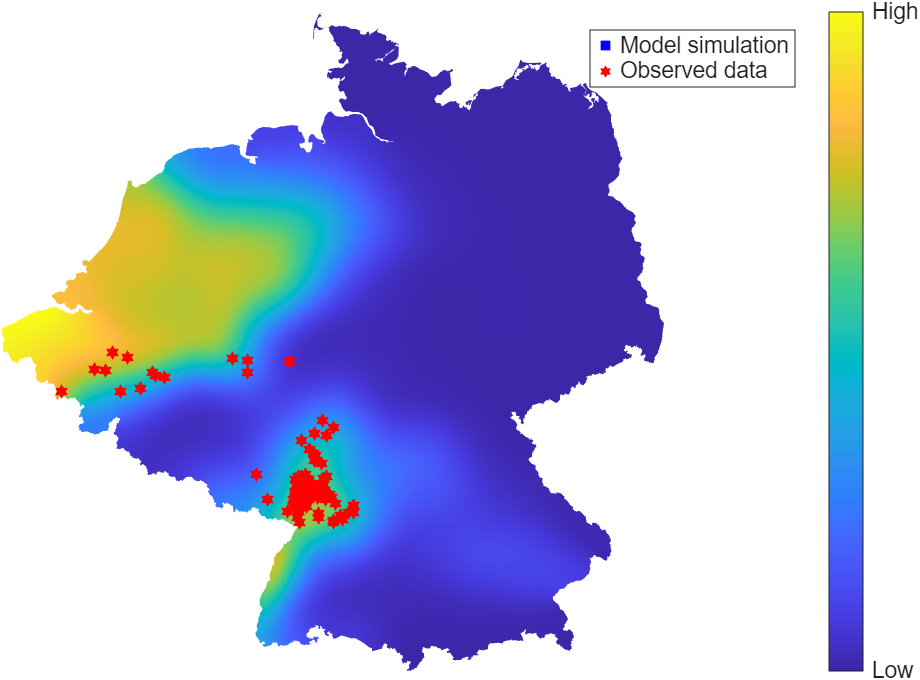}
			\caption{\small 2011-2014.}
			\label{f06}
		\end{subfigure}
		\hfill
		\begin{subfigure}{0.35\textwidth}
			\centering
			\includegraphics[width=\textwidth]{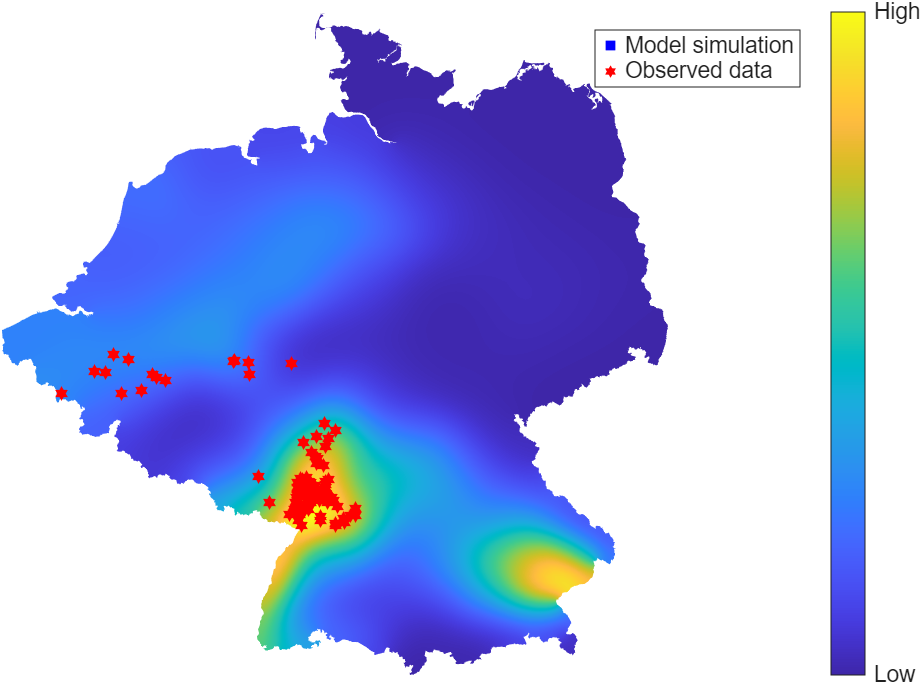}
			\caption{\small 2011-2015.}
			\label{f07}
		\end{subfigure}
		\hfill
		\begin{subfigure}{0.35\textwidth}
			\centering
			\includegraphics[width=\textwidth]{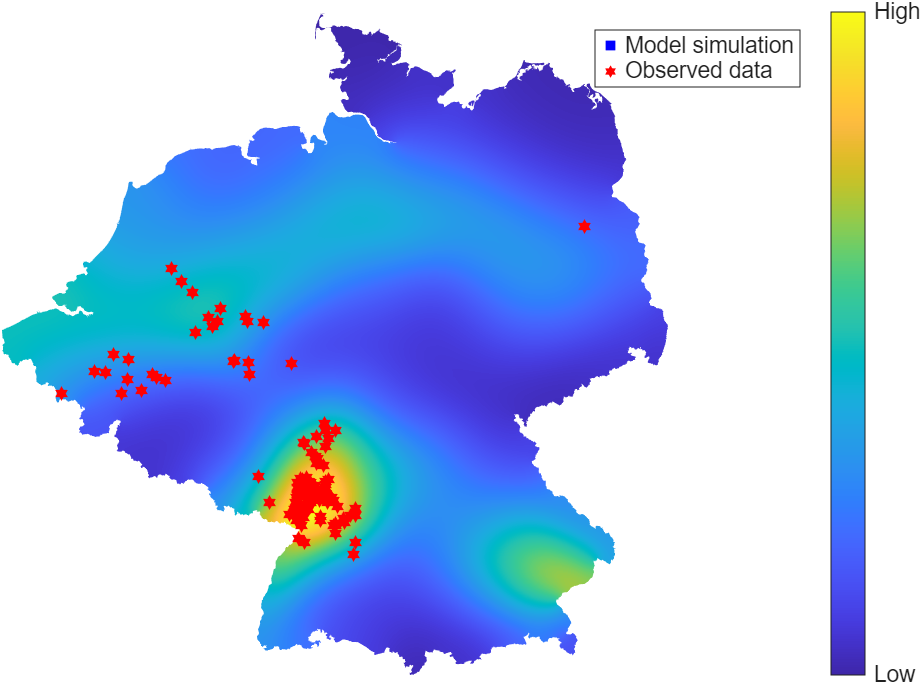}
			\caption{\small 2011-2016.}
			\label{f08}
		\end{subfigure}
		\hfill
		\begin{subfigure}{0.35\textwidth}
			\centering
			\includegraphics[width=\textwidth]{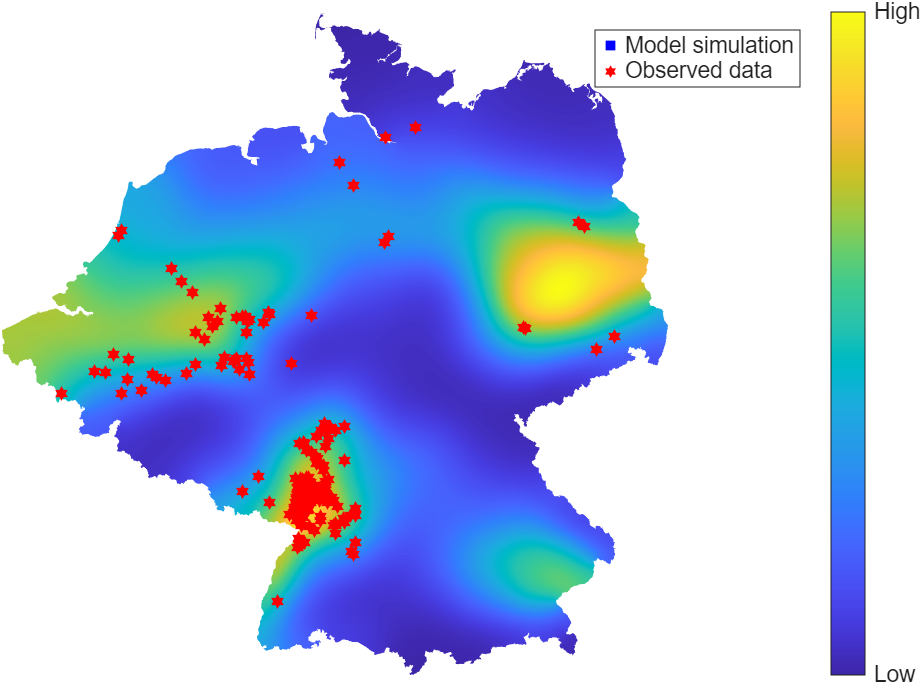}
			\caption{\small 2011-2017.}
			\label{f09}
		\end{subfigure}
		\hfill
		\begin{subfigure}{0.35\textwidth}
			\centering
			\includegraphics[width=\textwidth]{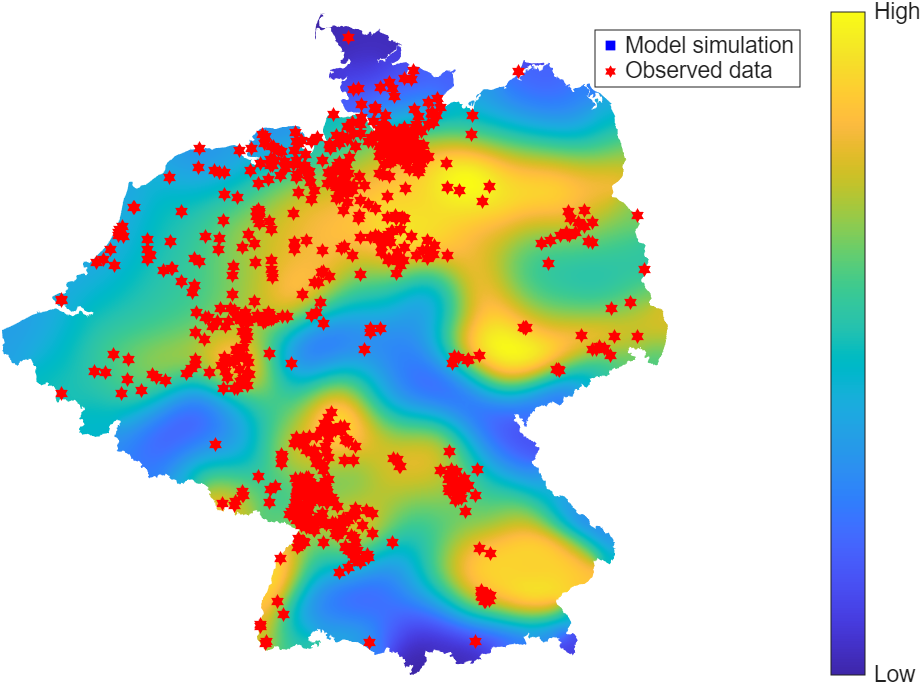}
			\caption{\small 2011-2018.}
			\label{f10}
		\end{subfigure}
		\hfill
		\caption{\small Model simulation compared to observed USUV data during the years 2011-2018.}
		\label{f2}
	\end{figure}

		\begin{figure}[H]
		\centering
		\begin{subfigure}{0.75\textwidth}
			\centering
			\includegraphics[width=\textwidth]{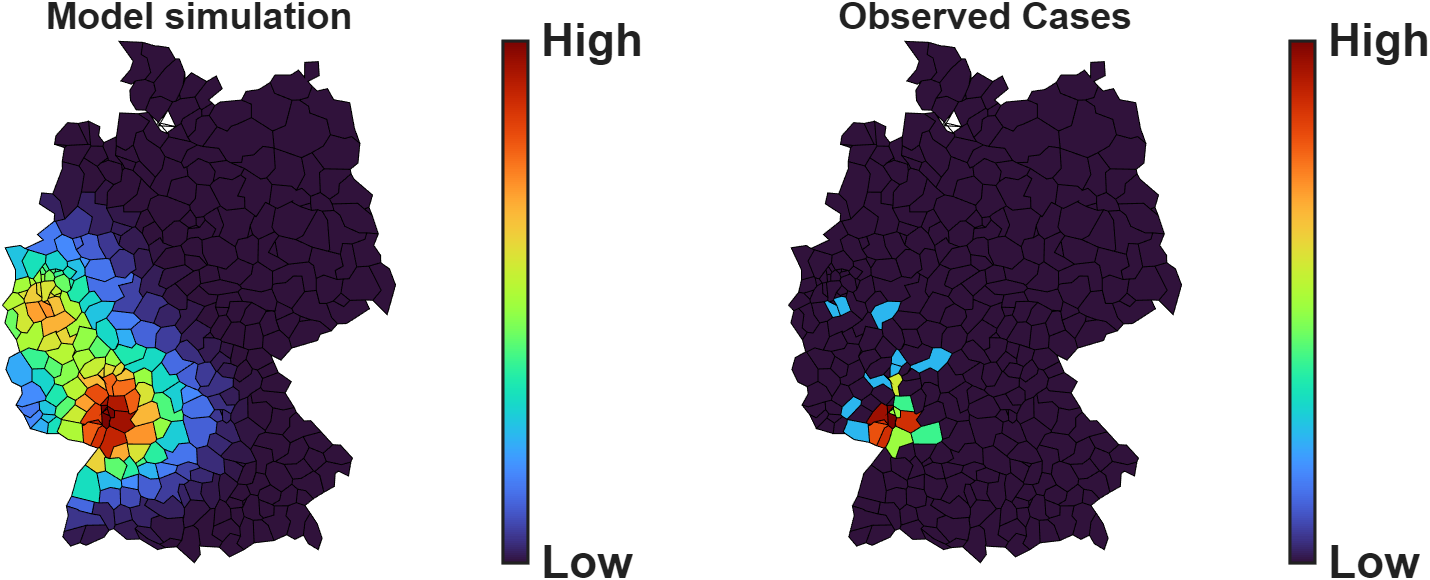}
			\caption{\small 2011-2012.}
			\label{l03}
		\end{subfigure}
		\hfill
		\begin{subfigure}{0.75\textwidth}
			\centering
			\includegraphics[width=\textwidth]{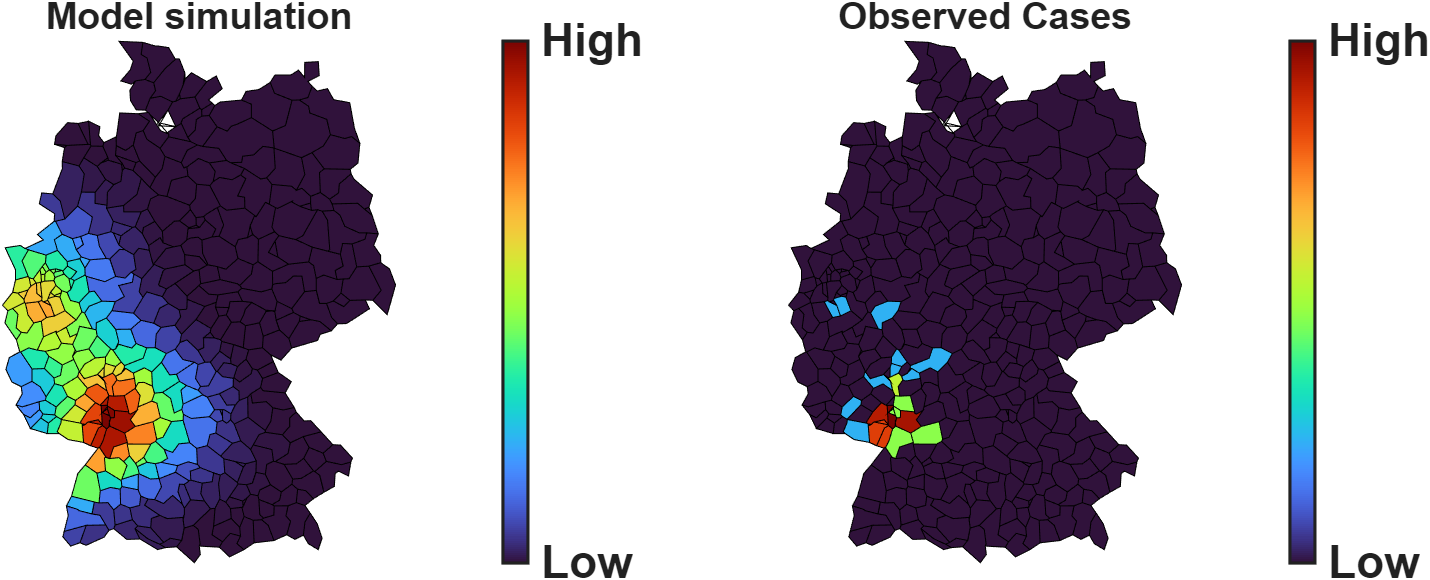}
			\caption{\small 2011-2013.}
			\label{l05}
		\end{subfigure}
		\hfill
		\begin{subfigure}{0.75\textwidth}
			\centering
			\includegraphics[width=\textwidth]{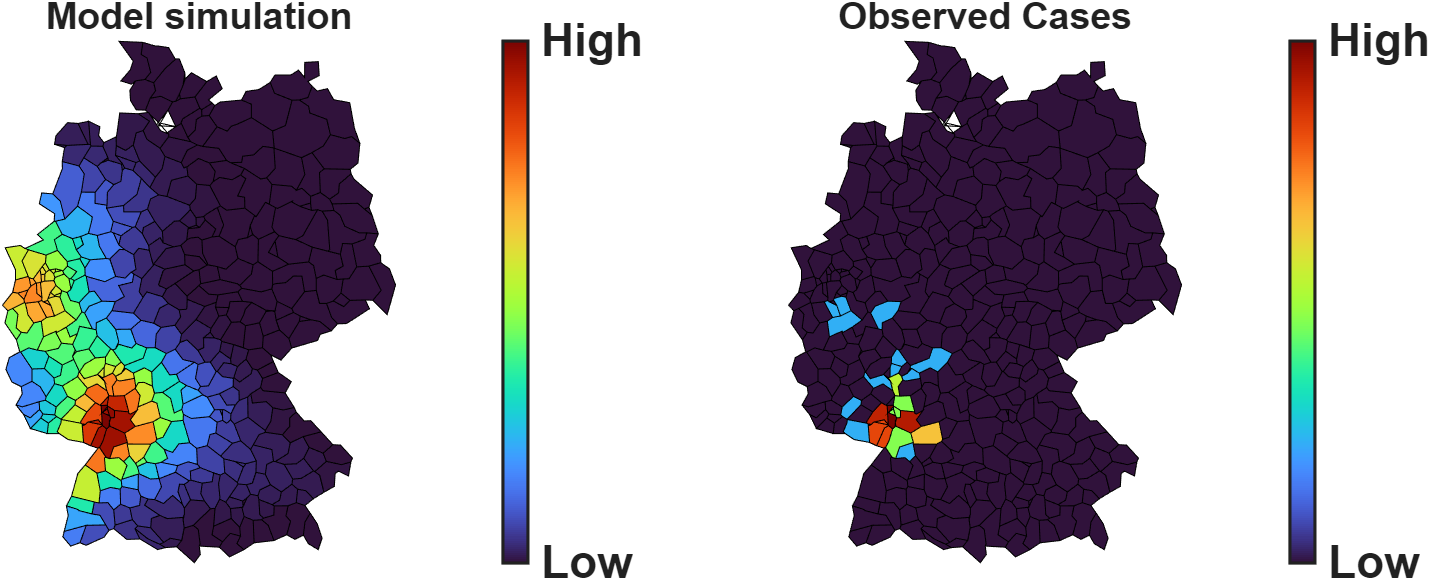}
			\caption{\small 2011-2014.}
			\label{l06}
		\end{subfigure}
		\hfill
		\begin{subfigure}{0.75\textwidth}
			\centering
			\includegraphics[width=\textwidth]{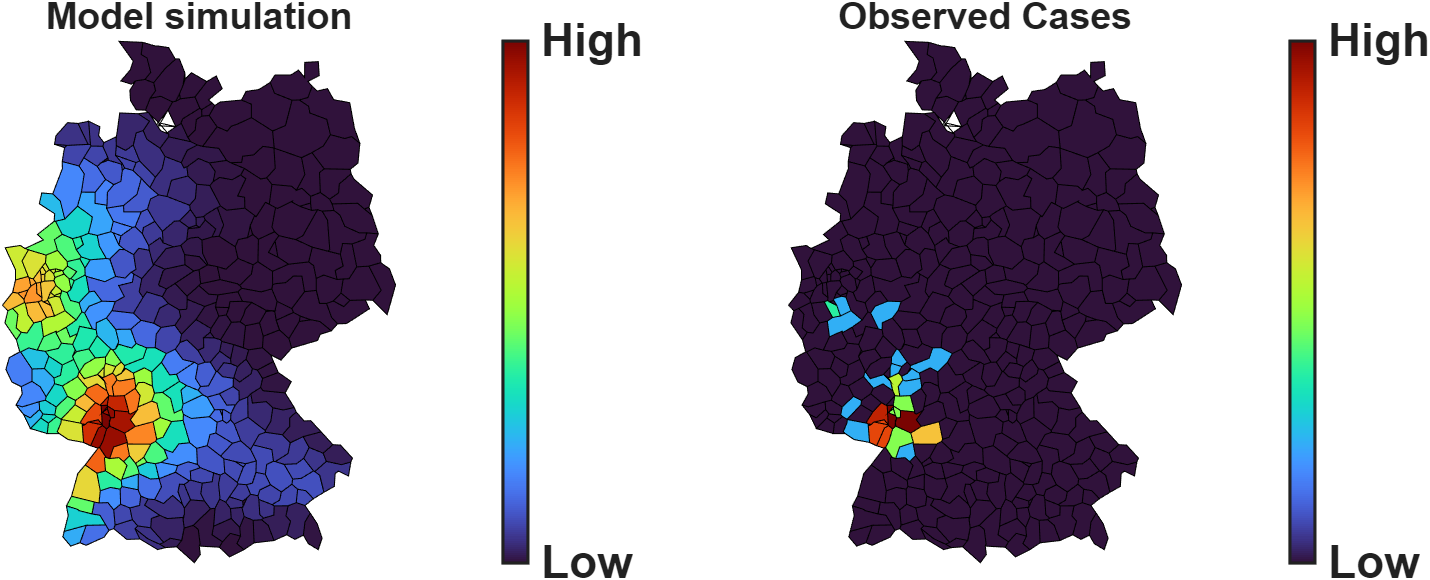}
			\caption{\small 2011-2015.}
			\label{l07}
		\end{subfigure}
		\hfill
		\caption{\small  Spatial comparison between simulated and observed USUV cases in Germany (2012-2018), on a $\log_{10}$ scale.}
		\label{l2}
	\end{figure}

		\begin{figure}[H]
		\centering
				\begin{subfigure}{0.75\textwidth}
			\centering
			\includegraphics[width=\textwidth]{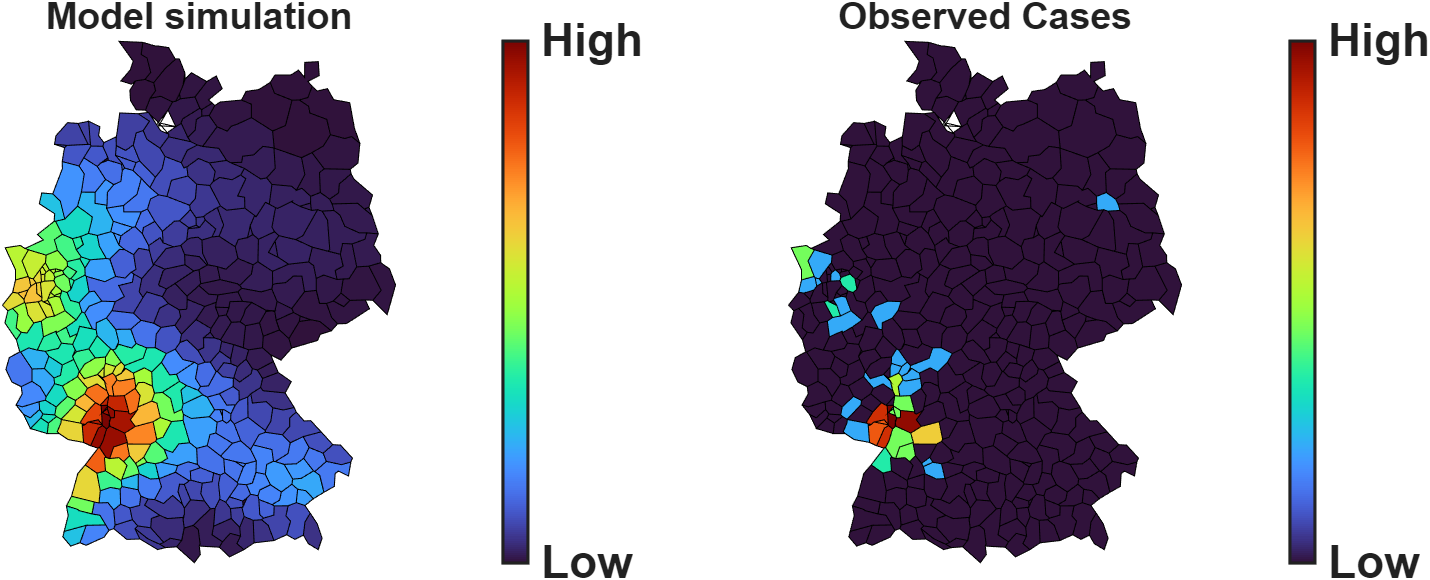}
			\caption{\small 2011-2016.}
			\label{l08}
		\end{subfigure}
		\hfill
		\begin{subfigure}{0.75\textwidth}
			\centering
			\includegraphics[width=\textwidth]{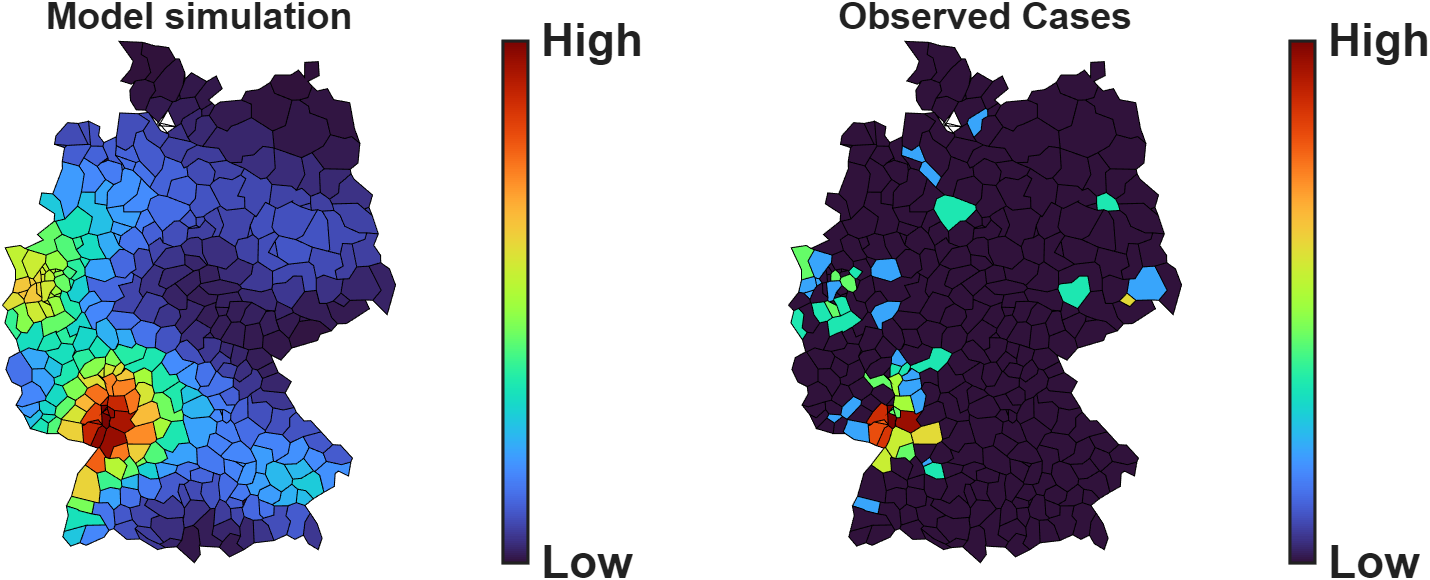}
			\caption{\small 2011-2017.}
			\label{l09}
		\end{subfigure}
					\hfill
		\begin{subfigure}{0.75\textwidth}
			\centering
			\includegraphics[width=\textwidth]{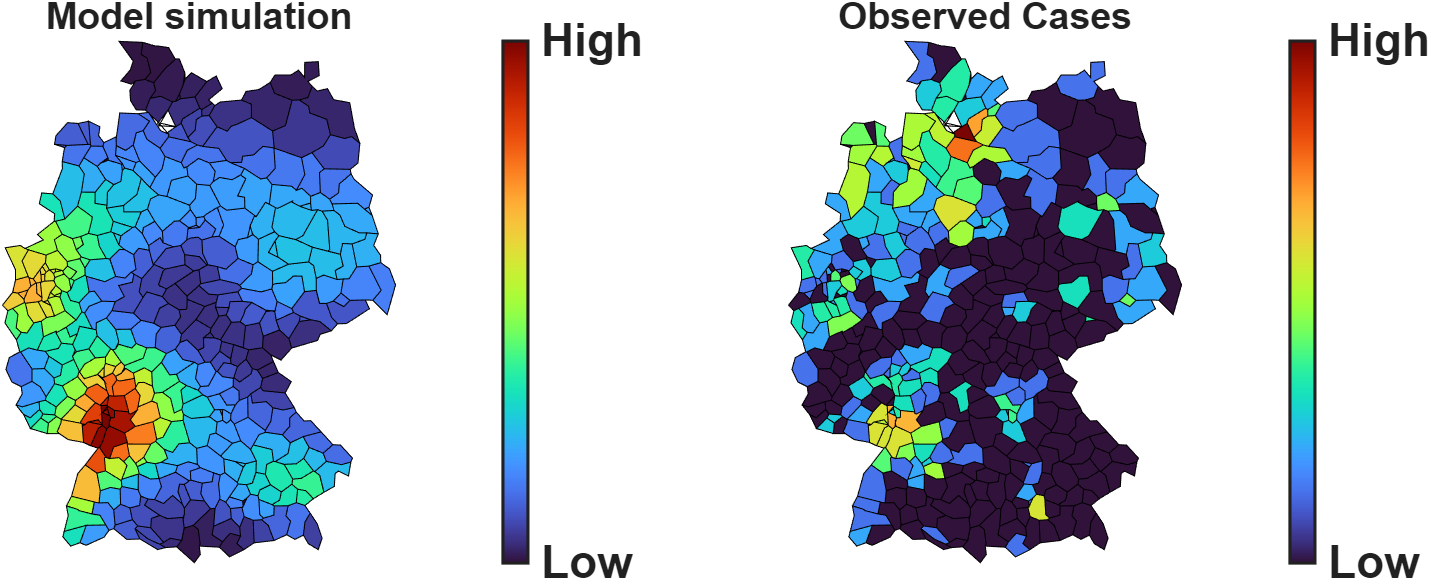}
			\caption{\small 2011-2018.}
			\label{l10}
		\end{subfigure}
		\hfill
		\caption{\small  Spatial comparison between simulated and observed USUV cases in Germany (2012-2018), on a $\log_{10}$ scale.}
		\label{k2}
	\end{figure}

		\begin{figure}[H]
		\centering
		\begin{subfigure}{0.45\textwidth}
			\centering
			\includegraphics[width=\textwidth]{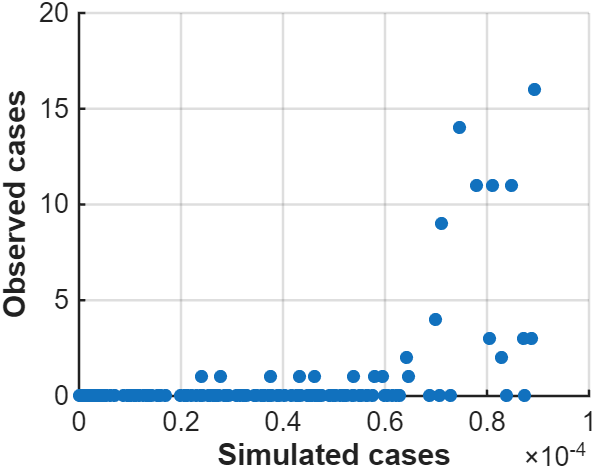}
			\caption{\small 2011-2012.}
			\label{g03}
		\end{subfigure}
		\hfill
		\begin{subfigure}{0.47\textwidth}
			\centering
			\includegraphics[width=\textwidth]{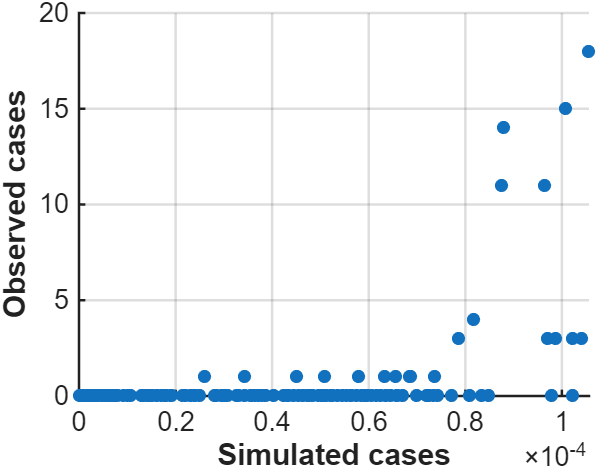}
			\caption{\small 2011-2013.}
			\label{g05}
		\end{subfigure}
		\hfill
		\begin{subfigure}{0.45\textwidth}
			\centering
			\includegraphics[width=\textwidth]{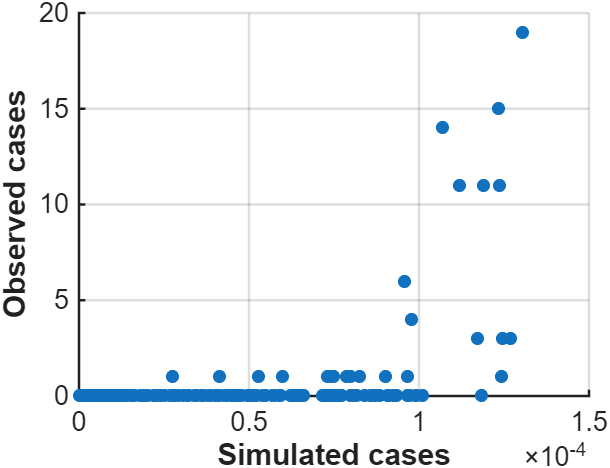}
			\caption{\small 2011-2014.}
			\label{g06}
		\end{subfigure}
		\hfill
		\begin{subfigure}{0.45\textwidth}
			\centering
			\includegraphics[width=\textwidth]{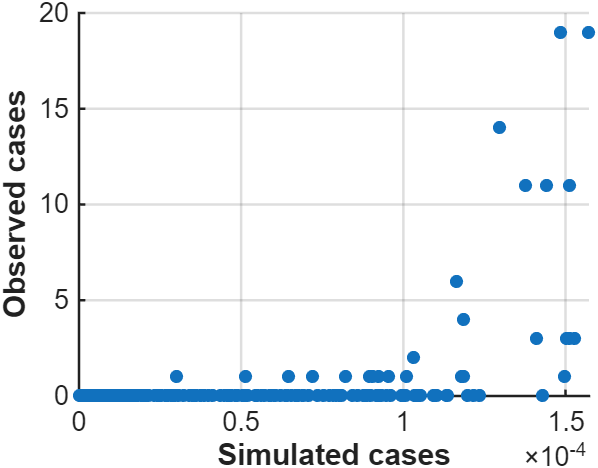}
			\caption{\small 2011-2015.}
			\label{g07}
		\end{subfigure}
		\hfill
		\begin{subfigure}{0.45\textwidth}
			\centering
			\includegraphics[width=\textwidth]{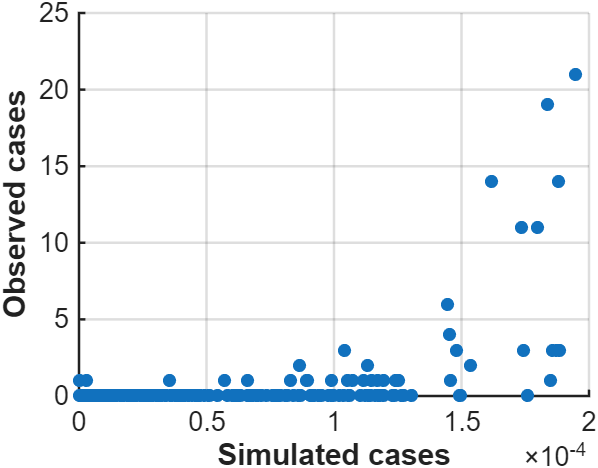}
			\caption{\small 2011-2016.}
			\label{g08}
		\end{subfigure}
		\hfill
		\begin{subfigure}{0.45\textwidth}
			\centering
			\includegraphics[width=\textwidth]{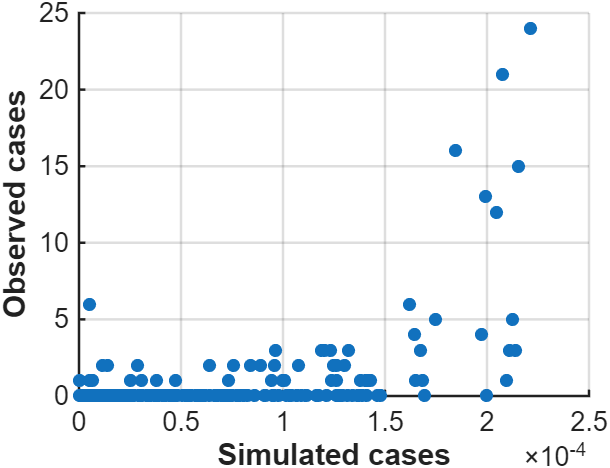}
			\caption{\small 2011-2017.}
			\label{g09}
		\end{subfigure}
		\hfill
		\begin{subfigure}{0.45\textwidth}
			\centering
			\includegraphics[width=\textwidth]{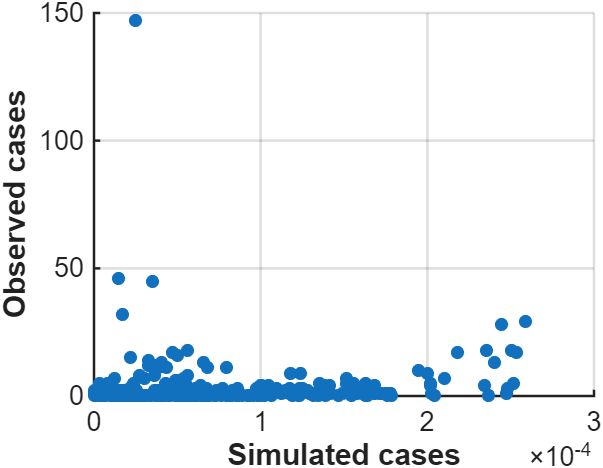}
			\caption{\small 2011-2018.}
			\label{g10}
		\end{subfigure}
		\hfill
		\caption{\small Scatter plots of simulated (x-axis) versus observed (y-axis) case numbers.}
		\label{g2}
	\end{figure}

\section{Discussions}
	
This study provides new insights into the spread of USUV in Germany, taking into account cross-border dynamics with the Netherlands, Belgium, and Luxembourg since its first detection in Germany in 2010. We developed a reaction-diffusion PDE model for mosquitoes and birds, incorporating temperature-dependent mosquito parameters. We proved mathematical properties such as non-negativity, uniqueness of solutions and the well-posedness of the model (Appendix \ref{mt}). The basic reproductive number $R_0(\boldsymbol{x})$ is used to assess low- and high-risk regions, with high $R_0(\boldsymbol{x})$ values indicating heightened risk. $R_0(\boldsymbol{x})$ values ranged between 0.3 and 2.2 (Figure \ref{fb}), and regions with higher values indicate a heightened risk of USUV transmission. Our $R_0(\boldsymbol{x})$ is primarily driven by temperature, thereby providing a spatially explicit transmission risk map for USUV, especially since \textit{Cx. pipiens} are ectothermic insects \citep{Ciota2014, vinogradova2000culex}. The $R_0(\boldsymbol{x})$ map aligned well with observed data and findings from \cite{Cheng2018}, identifying high-risk areas and further supporting the role of temperature in mapping USUV risk. 

Notably, $R_0(\boldsymbol{x})$ identifies areas around central Germany as having a lower risk for USUV circulation, whereas regions in the south-west (around the Rhine valley) are classified as high risk. This is mainly due to the region's favourable temperature for the establishment of USUV. Although other factors, such as landscape structure or reduced anthropogenic breeding habitats, may influence risk, the model's capacity to replicate observed patterns using only temperature suggests that temperature is a key variable, a conclusion supported by \cite{Rubel2008}.

Our PDE model successfully traced the spread pattern of USUV from Weinheim and Bonn, along with their surrounding regions, in 2011, and continued through 2018, when the virus was confirmed in almost all parts of Germany, parts of the Netherlands, and Belgium. USUV spread in a clockwise direction, entirely opposite to the spread pattern observed for WNV \citep{duve2025modeling}, despite both viruses sharing the same vector. Model simulations suggest that the differing spatial spread patterns of the two viruses mainly result from variations in their introduction pathways. Specifically, USUV was first introduced to southwest Germany, while WNV in Germany spread from the east. These distinct entry points influenced the subsequent direction and timing of spread, producing different spatial trajectories despite similar transmission mechanisms.

Although the spreading patterns of WNV and USUV differ, it is notable that both viruses spread along the same transmission corridor which is mainly shaped by regional temperatures (Figure \ref{f2} and \cite{duve2025modeling}). This result is supported by some studies that highlight the shared ecological niches between the two viruses \citep{Simonin2024}. Moreover, the transmission corridor is also in agreement with areas of high USUV suitability identified by \cite{Cheng2018} using a combination of mechanistic and the Maxent models. However, a PDE model further provides detailed information about the spreading patterns and speeds without being limited to suitability only. Because both viruses are now co-circulating in Germany \citep{Schopf2024} and the same temperature-driven corridor, surveillance and control measures can be concentrated in this zone to improve early season detection and interventions. Therefore, interventions aimed at controlling mosquito populations, such as habitat reduction and larval source management, may reduce the transmission of both viruses simultaneously. Consequently, vector control strategies implemented for WNV in Germany are likely to offer parallel benefits in limiting the spread of USUV.

A mosquito to bird ratio ranging between 30 and 50 was chosen as it produced a realistic spread pattern that biologically meaningful diffusion alone could not reproduce. Lower ratios failed to generate sufficient biting pressure to sustain transmission spread across the country, whereas higher ratios lead to unrealistically rapid spread. Given that bird sampling was not systematic across locations, the observed data probably contain sampling biases that might affect the true spatial pattern. This affected our model fit procedure, as we obtained low correlations in the model fit process. Furthermore, our modelling work can be improved by the use of spatially varying diffusion coefficients than will further support the spatial heterogeneity of the landscapes. A systematic method for estimating diffusion coefficients can further improve the uncertainties in the choice of diffusion coefficients.

\section*{CRediT authorship contribution statement}
	
\textbf{Pride Duve:} Conceptualization, Writing - Original draft, Formal analysis, Visualization;

\textbf{D\'aniel Cadar:} Investigation, Resources;

\textbf{Norbert Becker:} Investigation, Resources;

\textbf{Jonas Schmidt-Chanasit:} Investigation, Resources;

\textbf{ Felix Gregor Sauer:} Conceptualization, Supervision, Writing - Review \& Editing; 

\textbf{Renke Lühken:} Conceptualization, Supervision, Writing - Review \& Editing

\section*{Declaration of competing interests}
	
The authors declare that they have no known competing financial interests or personal relationships that could have appeared to influence the work reported in this paper.

\section*{Data availability}

The citizen science project data used in this study will be made freely available in Zenodo upon publication of this study.


	\newpage
	\section{Appendix}
	
	\subsection{Mathematical analysis of the model}\label{mt}
	
	\begin{theorem}
		System (\ref{xa}) has non-negative and unique solutions that are bounded in $\displaystyle [0,\infty).$
	\end{theorem}

	\begin{proof}
		
		We write system (\ref{xa}) together with the corresponding initial conditions, in the Banach space of continuous functions that include the boundary, $\displaystyle \mathcal{B}=C(\bar{\Omega}),$ as follows:
		

		\begin{equation*}\label{er}
			\begin{cases}
				\displaystyle \frac{\partial X(t,\boldsymbol{x})}{\partial t} =&\displaystyle \mathcal{L}X(t,\boldsymbol{x}) + f(X(t,\boldsymbol{x})),\; t>0,\\
				\displaystyle X(t,0) =&\displaystyle X_0\geq \boldmath{0}_{\mathbb{R}^{7}},
		\end{cases}\end{equation*}
		
		where 
		
		$$X(t,\boldsymbol{x})=\begin{bmatrix}
			S_V(t,\boldsymbol{x})\\
			E_V(t,\boldsymbol{x})\\
			I_V(t,\boldsymbol{x})\\
			S_{B}(t,\boldsymbol{x})\\
			E_{B}(t,\boldsymbol{x})\\
			I_{B}(t,\boldsymbol{x})\\
			R_{B}(t,\boldsymbol{x})\\
		\end{bmatrix},\; X(0,\boldsymbol{x})=\begin{bmatrix}
			S_V(0,\boldsymbol{x})\\
			E_V(0,\boldsymbol{x})\\
			I_V(0,\boldsymbol{x})\\
			S_{B}(0,\boldsymbol{x})\\
			E_{B}(0,\boldsymbol{x})\\
			I_{B}(0,\boldsymbol{x})\\
			R_{B}(0,\boldsymbol{x})\\
		\end{bmatrix}, \; \mathcal{L}X(t,\boldsymbol{x}) = \begin{bmatrix}
			D_{1}\Delta S_V(t,\boldsymbol{x})\\
			D_{1}\Delta E_V(t,\boldsymbol{x})\\
			D_{1}\Delta I_V(t,\boldsymbol{x})\\
			D_{2}\Delta S_{B}(t,\boldsymbol{x})\\
			D_{2}\Delta E_{B}(t,\boldsymbol{x})\\
			D_{2}\Delta I_{B}(t,\boldsymbol{x})\\
			D_{2}\Delta R_{B}(t,\boldsymbol{x})\\
		\end{bmatrix},$$
		
		and	
		
		\begin{eqnarray} \label{g1} 
			f =	\begin{cases}
				f_1=& \dis b_VN_V\left[1-\frac{N_V}{K_V}\right] -[\lambda_{BV}(T)+\mu_V(T)]S_V,\nonumber\\			
				f_2=& \lambda_{BV}(T)S_V-(\gamma_V(T)+\mu_V(T))E_V,\nonumber\\
				f_3=& \gamma_V(T)E_V-\mu_V(T)I_V,\nonumber\\
				f_4=& \dis b_{B}N_{B}\left[1-\frac{N_B}{K_B}\right]-\left[\lambda_{VB}(T)+\mu_{B}\right]S_{B},\nonumber\\
				f_5=& \lambda_{VB}(T)S_{B}-[\gamma_{B}+\mu_{B}]E_{B},\nonumber\\
				f_6=& \gamma_{B}E_{B}-[\alpha_{B}+\mu_{B}]I_{B},\nonumber\\
				f_7=& (1-\nu_{B})\alpha_{B} I_{B}-\mu_{B}R_{B}.\nonumber\\
			\end{cases}
		\end{eqnarray}

		Next, we show that $f$ is locally Lipschitz in $\mathcal{B}.$ Thus, we show that 
		$$\forall\; K\subset \mathcal{B},\; \exists L:\quad \|f(X_1)-f(X_2)\|_{\infty}\leq L\|X_1-X_2\|_{\infty},\;\forall X_1,X_2\in K,$$
		where $L$ is the Lipschitz constant. Setting $\dis \Lambda_{V}=b_VN_V\left[1-\frac{N_V}{K_V}\right]$ and $\dis \Lambda_{B}=b_BN_B\left[1-\frac{N_B}{K_B}\right]$, we observe that 
		$\displaystyle X_1-X_2=$
		$$
		\begin{bmatrix}
			\left[\Lambda_{V_1}-\Lambda_{V_2}\right]-\left[\lambda_{BV_1}S_{V_1}-\lambda_{BV_2}S_{V_2}\right]-\mu_V(S_{V_1}-S_{V_2})\\
			\left[\lambda_{BV_1}S_{V_1}-\lambda_{BV_2}S_{V_2}\right]-(\gamma_V+\mu_V)[E_{V_1}-E_{V_2}]\\
			\gamma_V[E_{V_1}-E_{V_2}]-\mu_V[I_{V_1}-I_{V_2}]\\
			[\Lambda_{B_1}-\Lambda_{B_2}]-(\lambda_{VB_1}S_{B_1}-\lambda_{VB_2}S_{B_2})-\mu_{B}[S_{B_1}-S_{B_2}]\\
			\lambda_{VB_1}S_{B_1}-\lambda_{VB_2}S_{B_2}-(\gamma_{B}+\mu_{B})[E_{B_1}-E_{B_2}]\\
			\gamma_{B}[E_{B_1}-E_{B_2}] - (\alpha_{B}+\mu_{B})[I_{B_1}-I_{B_2}]\\
			(1-\nu_{B})\alpha_{B}[I_{B_1}-I_{B_2}]-\mu_{B}[R_{B_1}-R_{B_2}]\\
		\end{bmatrix}.$$

		Simplifying, we obtain $\displaystyle \|f(X_1)-f(X_2)\|_{\infty}=$
		
		\begin{eqnarray*}
			&=& \displaystyle\sup_{ x\in\bar{\Omega}}|\left[\Lambda_{V_1}-\Lambda_{V_2}\right]-S_{V_1}[\lambda_{BV_1}-\lambda_{BV_2}]|\nonumber\\
			&\vee&\displaystyle\sup_{ x\in\bar{\Omega}}|-\lambda_{BV_2}[S_{V_1}-S_{V_2}]-\mu_V(S_{V_1}-S_{V_2})|\nonumber\\
			&\vee&\displaystyle\sup_{ x\in\bar{\Omega}} |\gamma_V[E_{V_1}-E_{V_2}]-\mu_V[I_{V_1}-I_{V_2}]|  \nonumber\\
			&\vee&\displaystyle\sup_{ x\in\bar{\Omega}} |[\Lambda_{B_1}-\Lambda_{B_2}]-\lambda_{VB_1}[S_{B_1}-S_{B_2}]-S_{B_2}[\lambda_{VB_1}-\lambda_{VB_2}]-\mu_{B}[S_{B_1}-S_{B_2}]|  \nonumber\\
			&\vee&\displaystyle\sup_{ x\in\bar{\Omega}} |\lambda_{VB_1}[S_{B_1}-S_{B_2}]+S_{B_2}[\lambda_{VB_1}-\lambda_{VB_2}]-(\gamma_{B}+\mu_{B})[E_{B_1}-E_{B_2}]| \nonumber\\
			&\vee&\displaystyle\sup_{ x\in\bar{\Omega}} |\gamma_{B}[E_{B_1}-E_{B_2}] - (\alpha_{B}+\mu_{B})[I_{B_1}-I_{B_2}]| \nonumber\\
			&\vee&\displaystyle\sup_{ x\in\bar{\Omega}} |(1-\nu_{B})\alpha_{B}[I_{B_1}-I_{B_2}]-\mu_{B}[R_{B_1}-R_{B_2}]|.\nonumber\\
		\end{eqnarray*}

		By the triangular inequality, and some simplifications, we arrive at
		
		$$\|f(X_1)-f(X_2)\|_{\infty}\leq \left(\mu_V+\nu_{B}\alpha_{B}\right)\|X_1-X_2\|_{\infty},$$
		
		and thus $f$ is locally Lipschitz in $\mathcal{B}.$ By [\citep{Capasso1993}, Theorem B.17], \citep{Mora1983}, and [\citep{Smoller1983}, Theorem 14.4], there exist a local smooth and unique solution of system (\ref{xa}) in $\displaystyle \Omega.$ We observe that system (\ref{xa}) can be written in the form of system (14.12) in the book \citep{Smoller1983}, together with initial data defined in system (14.13). By [\citep{Smoller1983}, Theorem (14.14)], the solutions of system (\ref{xa}) are always positive.	
	\end{proof}

	\subsection{Well-posedness of the model in a feasible region}
	
	\begin{theorem}
		system (\ref{xa}) is well-posed mathematically and biologically, and for any 
		
		$\dis \left(S_V(0,x),E_V(0,x),I_V(0,x),S_{B}(0,x), E_{B}(0,x), I_{B}(0,x), R_{B}(0,x)\right)\in \mathbb{X},$ system (\ref{xa}) admits a unique positive solution: $\dis \left(S_V(t,x),E_V(t,x),I_V(t,x),S_{B}(t,x),E_{B}(t,x), I_{B}(t,x), R_{B}(t,x)\right)\in \mathbb{X},$ satisfying:
		
		$\dis (S_V(t,x),E_V(t,x),I_V(t,x),S_{B}(t,x), E_{B}(t,x), I_{B}(t,x), R_{B}(t,x) \in C^{1,2}((0,\infty)\times\bar{\Omega}) \times C^{1,2}((0,\infty)\times\bar{\Omega}),\;\text{where}\; \mathbb{X}:=C(\bar{\Omega})\times C(\bar{\Omega}).$
		
		Moreover, there exist another constant $C_1>0$ independent of initial data, such that the solution $\dis (S_V(t,x),E_V(t,x),I_V(t,x),S_{B}(t,x), E_{B}(t,x), I_{B}(t,x), R_{B}(t,x)(t,x))$ satisfies:
		
		$\|S_V(t,x)\|_{L^{\infty}(\Omega)}+\|E_V(t,x)\|_{L^{\infty}(\Omega)} +\|I_V(t,x)\|_{L^{\infty}(\Omega)} +\|S_{B}(t,x)\|_{L^{\infty}(\Omega)} +\|E_{B}(t,x)\|_{L^{\infty}(\Omega)}+ \|I_{B}(t,x)\|_{L^{\infty}(\Omega)}+\|R_{B}(t,x)\|_{L^{\infty}(\Omega)} \leq C_1,\;\text{for all}\; t>T_0>0.$
	\end{theorem}

	\begin{proof}
		
		We follow the approach presented in \citep{Peng2012, Wang2022}. Using the regularity theory of parabolic PDEs \citep{Pao1993}, system (6) admits a unique non-negative classical solution 
		
		$$\dis (S_B(t,x), E_B(t,x), I_B(t,x), R_B(t,x)) \in C^{1,2}((0,T_m)\times\bar{\Omega}) \times C^{1,2}((0,T_m)\times\bar{\Omega}),$$
		
		where $T_m$ represents the maximal existence time of the solution. By the strong maximum principle \citep{Protter1984}, then $\dis S_B(t,x), E_B(t,x), I_B(t,x), R_B(t,x)$ are positive in $\dis (0,T_m)\times\bar{\Omega}.$
		
		Summing up the equations gives 
		
		$$\frac{\partial N_B}{\partial t}=\mathcal{D}_B\Delta N_{B} +b_{B}N_B\left[1-\frac{N_B}{K_B}\right]-\mu_BN_B-\alpha_B\nu_BI_B,$$
		
		subject to homogeneous Neumann boundary conditions $\dis \nabla N_B\cdot n=0$ on $\partial \Omega$ and non-negative initial data $N_B(x,0)=N_B(x)\geq 0.$ Integrating over $\dis \Omega,$ we get

		$$\displaystyle \frac{d}{dt}\int_{\Omega} N_B(x,t)\;dx= \int_{\Omega} \left[b_B-\mu_B\right]N_B\;dx-\int_{\Omega} b_B\frac{N_B}{K_B}\;dx \leq \bar{r}\int_{\Omega} N_B(x,t)\;dx,\; x\in\Omega,\; t> 0, \bar{r}\leq b_{B,max},$$

		yielding $$\int_{\Omega} N_B(x,t)\;dx \leq \int_{\Omega} e^{tb_{B,max}}N_{B0}(x)\;dx,\quad x\in \Omega,\; t\geq 0.$$
		
		Thus, $ \dis \|S_B(t,\cdot)\|_{L^1(\Omega)},$ $\dis \|E_B(t,\cdot)\|_{L^1(\Omega)},$ $\dis \|I_B(t,\cdot)\|_{L^1(\Omega)},$ and $\dis \|R_B(t,\cdot)\|_{L^1(\Omega)}$ are bounded for all $\dis 0<t<T_m.$ By the positivity of $S_B(t,\cdot),$ $E_B(t,\cdot),$ $I_B(t,\cdot),$ $R_B(t,\cdot),$ and [\citep{Peng2012}, Lemma (3.1)], with $\sigma=p_0=1,$ we conclude that there exist a positive constant $C_1$ that does not depend on initial data such that the solution $S_B, E_B, I_B, R_B$ satisfies
		$$\|S_B(t,\cdot)\|_{L^{\infty}(\Omega)} +  \|E_B(t,\cdot)\|_{L^{\infty}(\Omega)}+  \|I_B(t,\cdot)\|_{L^{\infty}(\Omega)} +  \|R_B(t,\cdot)\|_{L^{\infty}(\Omega)}\leq C_1,\;\forall t>T_0.$$
		
		Similar arguments can be made for the mosquito population, thus we conclude that system (6) is well-posed.
	\end{proof}

	\subsection{Matrices used in the PDEToolbox}\label{mat}

	$$c=[D_1; D_1; D_1; D_2; D_2; D_2; D_2; 0],$$ while 
	
	\begin{equation}\label{f}   
		f = \begin{cases}
			\dis  b_VN_V\left[1-\frac{N_V}{K_V}\right] \quad-\quad [\lambda_{BV}(T,\boldsymbol{x})+\mu_V(T,\boldsymbol{x})]S_V,\\	
			\lambda_{BV}(T,\boldsymbol{x})S_V\quad-\quad\left[\gamma_V(T,\boldsymbol{x})+\mu_V(T,\boldsymbol{x})\right]E_V,\\	
			\gamma_V(T,\boldsymbol{x})E_V\quad-\quad\mu_V(T,\boldsymbol{x})I_V,\\
			\dis b_{B}N_{B}\left[1-\frac{N_B}{K_B}\right]\quad-\quad\left[\lambda_{VB}(T,\boldsymbol{x})+\mu_{B}\right]S_{B},\\
			\lambda_{VB}(T,\boldsymbol{x})S_{B}\quad-\quad[\gamma_{B}+\mu_{B}]E_{B},\\
			\gamma_{B}E_{B}\quad-\quad[\alpha_{B}+\mu_{B}]I_{B},\\
			(1-\nu_{B})\alpha_{B} I_{B}\quad-\quad\mu_{B}R_{B},\\
			\alpha_{B}\nu_{B} I_{B}.\end{cases}
	\end{equation}

	\bibliographystyle{plainnat}  
	
	\newpage
	\bibliography{References}      
	
\end{document}